\documentclass[10pt]{article}

\usepackage{amsmath, amsthm, amsfonts, amssymb, graphicx, mathrsfs }
\usepackage{hyperref}
\usepackage{subfig}
\usepackage{caption}

\newcommand{\matice}[1]{\left( \begin{array}{cc} #1 \end{array} \right)}
\newcommand{\eq}[1]{\begin{equation} #1 \end{equation}}

\newcommand{\N}{\mathbb{N}}
\newcommand{\R}{{\mathbb{R}}}
\newcommand{\C}{{\mathbb{C}}}
\newcommand{\Z}{{\mathbb{Z}}}
\newcommand{\T}{{\mathcal{T}}}
\newcommand{\dd}{{{\rm d}}}
\newcommand{\ii}{{\rm i}}
\renewcommand{\P}{{\mathcal{P}}}
\newcommand{\PT}{{\mathcal{PT}}}

\newcommand{\Dom}{{\rm{Dom}\,}}

\renewcommand{\Re}{\text{\rm Re}\,}
\renewcommand{\Im}{\text{\rm Im}\,}
\newcommand{\cf}{\emph{cf.}}
\newcommand{\diag}{{\rm diag}}
\newcommand{\ie}{{\emph{i.e.}}}
\newcommand{\eg}{{\emph{e.g.}}}
\newcommand{\Jj}{J_1}
\newcommand{\Jd}{J_2}
\newcommand{\I}{{\rm II}}
\newcommand{\II}{{\rm I}}
\newcommand{\dist}{\mathrm{dist}}

\begin{document}

\numberwithin{equation}{section}
\newtheorem{theorem}{Theorem}[section]
\newtheorem{lemma}{Lemma}[section]
\newtheorem{proposition}{Proposition}[section]
\newtheorem{corollary}{Corollary}[section]
\renewcommand{\proofname}{Proof}
\theoremstyle{remark}
\newtheorem{remark}{Remark}[section]


\title{$\PT$-symmetric models in curved manifolds}
\author{David Krej\v ci\v r\'ik$^{a}$ and Petr Siegl$^{a,b,c}$}
\date{
\small
\emph{
\begin{quote}
\begin{itemize}
\item[$a)$] Department of Theoretical Physics, Nuclear Physics Institute, Czech Academy of Sciences, \v Re\v z, Czech Republic \\
E-mail: david@ujf.cas.cz, siegl@ujf.cas.cz
\item[$b)$] Faculty of Nuclear Sciences and Physical Engineering, Czech Technical University in Prague, Prague,  Czech Republic
\item[$c)$] Laboratoire Astroparticule et Cosmologie, Universit\'e Paris 7, Paris, France
\end{itemize}
\end{quote}
}
\medskip
11 January 2010
}

\maketitle

\begin{abstract}
\noindent
We consider the Laplace-Beltrami operator in tubular neighbourhoods
of curves on two-dimensional Riemannian manifolds,
subject to non-Hermit\-ian parity and time preserving boundary conditions.
We are interested in the interplay between the geometry and spectrum.
After introducing a suitable Hilbert space framework
in the general situation,
which enables us to realize the Laplace-Beltrami operator
as an m-sectorial operator,
we focus on solvable models defined on manifolds of constant curvature.
In some situations, notably for non-Hermitian Robin-type boundary conditions,
we are able to prove either the reality
of the spectrum or the existence of complex conjugate pairs of eigenvalues,
and establish similarity of the non-Hermitian m-sectorial operators
to normal or self-adjoint operators.
The study is illustrated by numerical computations.
\\

\\


\end{abstract}

\newpage
\tableofcontents

\newpage
\section{Introduction}

Many systems in Nature can be under first approximation
described by linear second order differential equations,
such as the wave, heat or Schr\"odinger equation.
The common denominator of them is the Helmholtz equation
describing the stationary regime
and leading to the spectral study of the Laplace operator.
Already from the mathematical point of view,
it is important to understand the influence of the geometry
to the spectrum of the Laplacian, subject to various
types of boundary conditions, and vice versa,
to characterize geometric and boundary interface
properties from given spectral data.

In this paper, we are interested in the interplay between
the curvature of the ambient space and the spectrum of the Laplacian
subjected to a special class of non-Hermitian boundary conditions.
We choose probably the simplest non-trivial model,
\ie, the spectral problem corresponding to the equation
\begin{equation}\label{Helmholtz}
  -\Delta \psi = \lambda \psi
  \qquad\mbox{in}\qquad
  \Omega \,,
\end{equation}
where~$\lambda$ is a spectral parameter,
$\Omega$~is an $a$-tubular neighbourhood
of a closed curve~$\Gamma$ (usually a geodesic)
in a two-dimensional Riemannian manifold~$\mathcal{A}$
(not necessarily embedded in~$\R^3$), \ie,
\begin{equation}
  \Omega := \{
  x \in \mathcal{A} \ | \
  \dist(x,\Gamma) < a
  \}
  \,,
\end{equation}
and~$-\Delta$ is the associated Laplace-Beltrami operator.
The boundary conditions we consider
are general `parity and time preserving' boundary conditions
introduced in Section~\ref{bound.con} below;
a special example is given by the non-Hermitian Robin-type
boundary conditions
\begin{equation}\label{Robin}
  \frac{\partial\psi}{\partial n} + i \alpha \psi = 0
  \qquad\mbox{on}\qquad
  \partial\Omega \,,
\end{equation}
where~$n$ is the curve normal translated by geodesics
orthogonal to~$\Gamma$ and~$\alpha$ is a real-valued function.

The Schr\"odinger equation
in tubular neighbourhoods of submanifolds
of curved Riemannian manifolds
has been extensively studied in the context
of quantum waveguides and molecular dynamics
(\cf~\cite{Wachsmuth-2009} for a recent
mathematical paper with many references).
Here the confinement to a vicinity of the submanifold
is usually modelled by constraining potentials~\cite{Mitchell-2001-63,Wachsmuth-2009}
or Dirichlet boundary conditions~\cite{Clark-1996-29,Krejcirik-2003-45,Krejcirik-2006-2006}.

Note that, on the contrary,
the non-Hermitian nature of boundary conditions~\eqref{Robin}
enables one to model a leak/supply of energy
from/into the subsystem~$\Omega$,
since the probability current does not vanishes
on the boundary~$\partial\Omega$ unless $\alpha=0$.
In fact, non-Hermitian boundary conditions of the type~\eqref{Robin}
has been considered in \cite{Kaiser-2003-252,
Kaiser-2003-45, Kaiser-2002-43}
to model open (dissipative) quantum systems.
One also arrives at~\eqref{Robin} when transforming a scattering problem
to a (non-linear) spectral one~\cite[Ex.~9.2.4]{Davies-2007}.
Finally, let us observe that
Robin boundary conditions are known under the term
impedance boundary conditions in classical electromagnetism,
where they are conventionally used to approximate very thin layers
\cite{Bouchitte-1989-24, Engquist-1993, Bendali-1996-56}.

Our primary motivation to consider
the spectral problem~\eqref{Helmholtz},~\eqref{Robin}
comes from the so-called `$\mathcal{PT}$-symmetric quantum mechanics'
originated by the paper~\cite{bender-1998-80},
where the authors discussed a class of Schr\"odinger operators~$H$
in $L^2(\R)$ whose spectrum is real in spite of the fact
that their potentials are complex.
They argued that the rather surprising reality of the spectrum
follows from the $\mathcal{PT}$-symmetry property:
\begin{equation}\label{PT}
  [H,\mathcal{PT}] = 0 \,.
\end{equation}
Here the `parity'~$\mathcal{P}$
and `time reversal'~$\mathcal{T}$ operators
are defined by $(\mathcal{P}\psi)(x):=\psi(-x)$
and $\mathcal{T}\psi:=\overline{\psi}$.
It is important to emphasize that~$\mathcal{T}$
is an antilinear operator and that~\eqref{PT}
is neither sufficient nor necessary condition
to ensure the reality of the spectrum of~$H$.

Nevertheless, later on it was observed in
\cite{bender2002-89,Mostafazadeh-2002-43a,Mostafazadeh-2002-43b,Mostafazadeh-2002-43c}
that if the spectrum of a $\mathcal{PT}$-symmetric
operator~$H$ in a Hilbert space~$\mathcal{H}$ is indeed real
(and some further hypotheses are satisfied)
the condition~\eqref{PT} actually implies
that~$H$ is `quasi-Hermitian'~\cite{Scholtz-1992-213},
\ie, there exists a bounded invertible positive operator~$\Theta$ with bounded inverse,
called `metric', such that
\begin{equation}\label{quasi}
  H^* = \Theta^{-1} H \Theta \,.
\end{equation}
In other words, $H$~is similar to a self-adjoint operator
for which a conventional quantum-mechanical interpretation makes sense.
We refer to recent reviews~\cite{Bender2007-70,Mostafazadeh2008-review} and
proceedings \cite{jain-2009-73,Andrianov-2009-5,Fring-2008-41} for further information and references about
the concept of $\mathcal{PT}$-symmetry.

In addition to the potential quantum-mechanical interpretation,
we would like to mention the relevance of $\PT$-symmetric
operators in view of their recent study in the context of
superconductivity \cite{Rubinstein-2007-99,Rubinstein-2010-195},
electromagnetism \cite{Ruschhaupt-2005-38,Klaiman-2008-101}
and fluid dynamics
\cite{Chugunova-2008-342,Davies_2007,Weir-2009-22,Boulton-Levitin-Marletta}.

A suitable mathematical framework to analyse
$\mathcal{PT}$-symmetric Hamiltonians is either the theory
of self-adjoint operators in Krein spaces
\cite{Langer-2004-54,Jacob-2008-8}
or the $J$-self-adjointness~\cite{borisov-2007}.
The latter means that there exists an antilinear involution~$J$
such that
\begin{equation}\label{J}
  H^* = J H J \,.
\end{equation}
The concept~\eqref{J} is not restricted to functional Hilbert spaces
and it turns out that the majority of $\mathcal{PT}$-symmetric
Hamiltonians existing in the literature are indeed $J$-self-adjoint.
In general, however, the properties~\eqref{PT},
\eqref{quasi} and~\eqref{J} are all unrelated \cite{Siegl-MT,Siegl-2009-73}.

Summing up, given a non-Hermitian operator~$H$ satisfying~\eqref{PT},
two fundamental questions arises.
First,
\begin{enumerate}
\item
is the spectrum of~$H$ real?
\end{enumerate}
Second, if the answer to the previous question is positive,
\begin{enumerate}
\setcounter{enumi}{1} \item
does there exist a metric~$\Theta$ satisfying~\eqref{quasi}?
\end{enumerate}
It turns out that the questions constitute a difficult
problem in the theory of non-self-adjoint operators.

For this reason, one of the present authors and his coauthors
proposed in~\cite{krejcirik-2006-39} (see also~\cite{krejcirik-2008-41a})
an elementary one-dimensional $\mathcal{PT}$-symmetric Hamiltonian,
for which the spectrum and metric are explicitly computable.
The simplicity of the Hamiltonian consists in the fact that it acts
as the Hamiltonian of a free particle in a box
and the non-Hermitian interaction is introduced
via the Robin-type boundary conditions~\eqref{Robin} only.
The model was later generalized to a two-dimensional waveguide in~\cite{borisov-2007},
where the variable coupling in the boundary conditions
is responsible for existence of real
(or complex conjugate pairs of) eigenvalues
outside the essential spectrum (see also~\cite{krejcirik-2008-41}).

In this paper we continue the generalization of the models
of~\cite{krejcirik-2006-39,borisov-2007} to curved Riemannian manifolds.
This leads to a new, large class of $\mathcal{PT}$-symmetric Hamiltonians.
Our main goal is to study the effect of curvature
on the spectrum, namely the existence/absence
of non-real eigenvalues and the metric.

The organization of this paper is as follows.

In the following Section~\ref{Sec.def},
we introduce our model in a full generality,
in the sense that the ambient geometry and boundary interaction
of the spectral problem~\eqref{Helmholtz}
are described by quite arbitrary (non-constant and non-symmetric) functions.
Our main strategy to deal with the curved geometry
is based on the usage of Fermi coordinates.

In Section~\ref{Sec.form}, we use the framework of sesquilinear
forms to define the Laplace-Beltrami operator appearing in~\eqref{Helmholtz}
as a (closed) m-sectorial operator in the Hilbert space $L^2(\Omega)$.
We also explicitly determine the operator domain
if the assumptions about the geometry and boundary-coupling functions
are naturally strengthen. Moreover, we find conditions
about the geometry under which the operator becomes $\PT$-symmetric
(and $\mathcal{T}$-self-adjoint).

In order to study the effects of curvature on the spectrum,
in Section~\ref{Sec.constant} we focus on solvable models.
Assuming that the curvature and boundary-coupling functions are constant,
the eigenvalue problem can be reduced to the investigation
of (infinitely many) one-dimensional differential operators
with $\PT$-symmetric boundary conditions.
Here the previous results \cite{krejcirik-2006-39,krejcirik-2008-41a}
and the general theory of boundary conditions
for differential operators \cite{Naimark1967-LDOP1,Naimark1968-LDOP2}
are appropriate and helpful.
In particular, since the $\PT$-symmetric boundary conditions are
(except one case excluded here by assumption) strongly regular ones,
it is possible to show that the studied one-dimensional operators
are `generically' similar to self-adjoint or normal operators.
However, it remains to decide whether this is true for their infinite sum,
\ie, for the original two-dimensional Laplace-Beltrami operator.
To answer this in affirmatively,
it turns out that the $J$-self-adjoint formulation of $\PT$-symmetry
(\cf~the text around~\eqref{J}) is fundamental,
with~$J=\T$ playing the role of antilinear involution.
The properties of the solvable models
are illustrated by a numerical analysis of their spectra.

The paper is concluded by Section~\ref{Sec.end}
where possible directions of the future research are mentioned.

\section{Definition of the model}\label{Sec.def}

We use the quantum-mechanical framework to describe our model.

\subsection{The configuration space}

We assume that the ambient space of a quantum particle
is a connected complete two-dimensional
Riemannian manifold~$\mathcal{A}$ of class~$C^2$
(not necessarily embedded in the Euclidean space~$\R^3$).
Furthermore, we suppose that the Gauss curvature~$K$ of~$\mathcal{A}$ is continuous,
which holds under the additional assumption that~$\mathcal{A}$
is either of class~$C^3$ or it is embedded in~$\R^3$.

On the manifold, we consider a $C^2$-smooth unit-speed embedded curve
$\Gamma:[-l,l]\to\mathcal{A}$, with $l > 0$.
Since~$\Gamma$ is parameterized by arc length,
the derivative $T:=\dot\Gamma$ is the unit tangent vector of~$\Gamma$.
Let~$N$ be the unit normal vector of~$\Gamma$
which is uniquely determined as the $C^1$-smooth mapping
from $[-l,l]$ to the tangent bundle of~$\mathcal{A}$
by requiring that~$N(s)$ is orthogonal to~$T(s)$
and that $\{T(s),N(s)\}$ is positively oriented for all $s\in[-l,l]$
(\cf~\cite[Sec.~7.B]{Spivak-1975}).
We denote by~$\kappa$ the corresponding curvature of~$\Gamma$
defined by the Frenet formula
$
  \nabla_T T = \kappa N
$,
where~$\nabla$ stands for the covariant derivative in~$\mathcal{A}$.
We note that the sign of~$\kappa$ is uniquely determined
up to the re-parametrization $s \mapsto -s$ of the curve~$\Gamma$
and that~$\kappa$ coincides with the geodesic curvature of~$\Gamma$
if~$\mathcal{A}$ is embedded in~$\R^3$.

The feature of our model is that the particle is assumed
to be `confined' to an $a$-tubular neighbourhood~$\Omega$ of~$\Gamma$,
with $a>0$.
$\Omega$~can be visualized as the set of points~$q$ in~$\mathcal{A}$
for which there exists a geodesic of length less than~$a$
from~$q$ meeting~$\Gamma$ orthogonally.
More precisely, we introduce a mapping~$\mathcal{L}$
from the rectangle
\begin{equation}\label{Omega0}
 \Omega_0:=(-l,l)\times(-a,a) \equiv \Jj \times \Jd
\end{equation}
(considered as a subset of the tangent bundle of~$\mathcal{A}$)
to the manifold~$\mathcal{A}$ by setting
\begin{equation}\label{exp}
  \mathcal{L}(x_1,x_2) := \exp_{\Gamma(x_1)}(N(x_1) \,x_2)
  \,,
\end{equation}
where $\exp_q$ is the exponential map of~$\mathcal{A}$ at $q\in\mathcal{A}$,
and define
\begin{equation}\label{image}
  \Omega := \mathcal{L}(\Omega_0)
  \,.
\end{equation}
Note that $x_1\mapsto\mathcal{L}(x_1,x_2)$ traces the curves
parallel to~$\Gamma$ at a fixed distance~$|x_2|$,
while the curve~$x_2\mapsto\mathcal{L}(x_1,x_2)$
is a geodesic orthogonal to~$\Gamma$ for any fixed~$x_1$.
See Figure~\ref{ptstrip}.

\begin{figure}[ht]
    \centering
\subfloat{\includegraphics[width =0.5\textwidth]{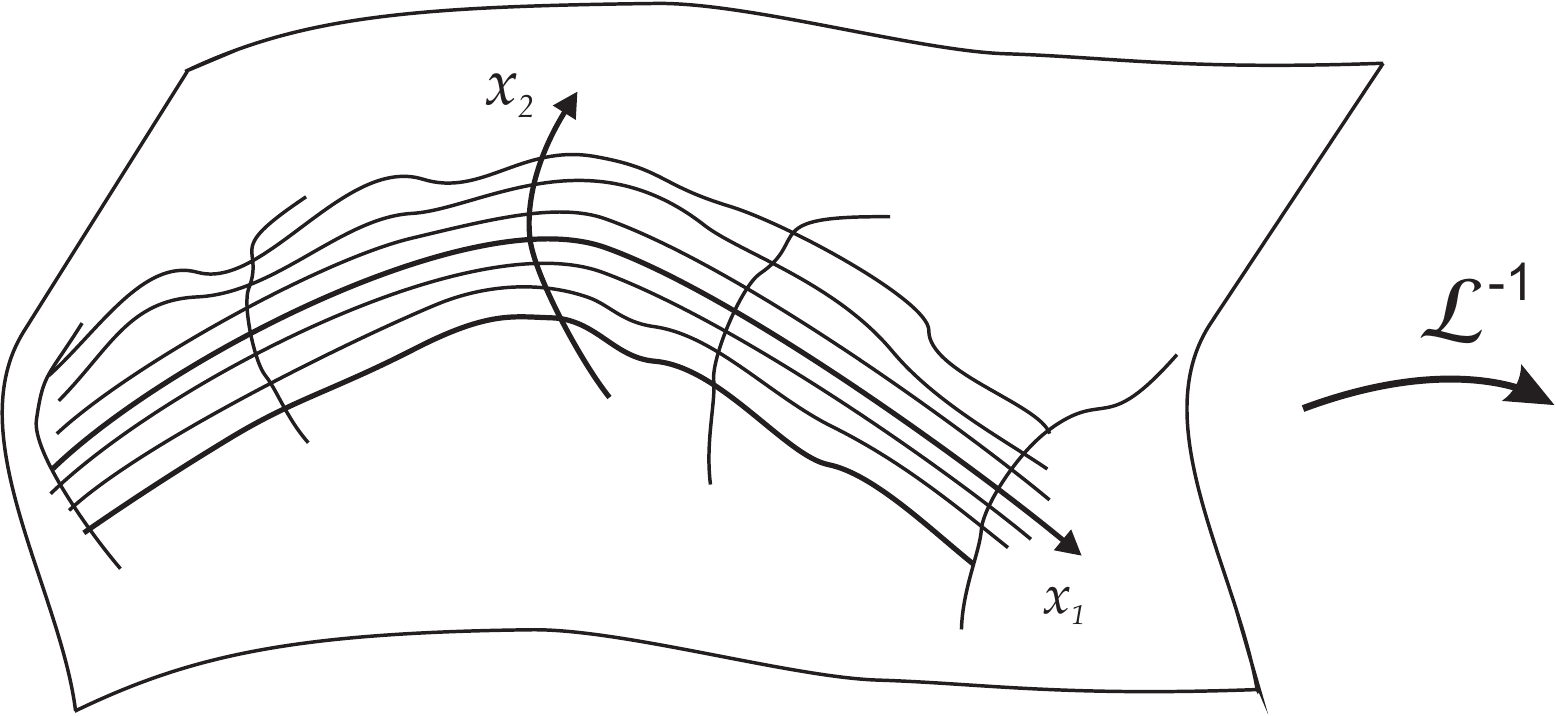}}
\subfloat{\includegraphics[width =0.5\textwidth]{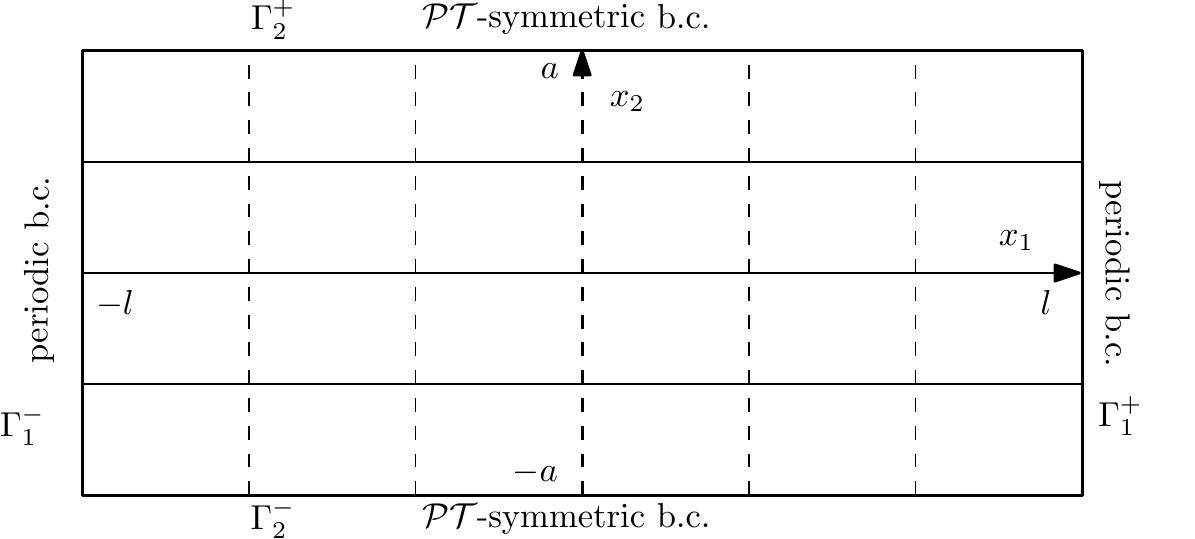}}
    \caption{Strip and boundary conditions}
    \label{ptstrip}
\end{figure}

\subsection{The Fermi coordinates}

Throughout the paper we make the hypothesis that
\begin{equation}\label{Ass.basic}
  \fbox{\ $\mathcal{L} : \Omega_0 \to \Omega$ is a diffeomorphism.}
\end{equation}
Since~$\Gamma$ is compact, \eqref{Ass.basic}
can always be achieved for sufficiently small~$a$
(\cf~\cite[Sec.~3.1]{Gray}).
Consequently, $\mathcal{L}$~induces a Riemannian metric~$G$ on~$\Omega_0$,
and we can identify the tubular neighbourhood $\Omega \subset \mathcal{A}$
with the Riemannian manifold $(\Omega_0,G)$.
In other words, $\Omega$~can be conveniently parameterized
via the (Fermi or geodesic parallel) ``coordinates" $(x_1,x_2)$
determined by~\eqref{exp}.
We refer to~\cite[Sec.~2]{Gray} and~\cite{Hartman_1964}
for the notion and properties of Fermi coordinates.
In particular, it follows by the generalized Gauss lemma
that the metric acquires the diagonal form:
\begin{equation}\label{metric}
  G =
  \begin{pmatrix}
    f^2 & 0 \\
    0   & 1
  \end{pmatrix}
  \,,
\end{equation}
where~$f$ is continuous and has continuous partial derivatives
$\partial_2 f$, $\partial_2^2 f$ satisfying the Jacobi equation
\begin{equation}\label{Jacobi}
  \partial_2^2 f + \,K f = 0
  \qquad\textrm{with}\qquad\left\{
  \begin{aligned}
    f(\cdot,0) &= 1 \,, \\
    \partial_2 f(\cdot,0) &= -\,\kappa \,.
  \end{aligned}
  \right.
\end{equation}
Here~$K$ is considered as a function of the Fermi coordinates~$(x_1,x_2)$.

\subsection{The Hamiltonian}

We identify the Hamiltonian~$H$ of the quantum particle in~$\Omega$
with the Laplace-Beltrami operator~$-\Delta_G$
in the Riemannian manifold $(\Omega_0,G)$,
subject to a special class of non-self-adjoint boundary conditions.

\subsubsection{The action of the Hamiltonian}

Denoting by $G^{ij}$ the coefficients of the inverse metric~$G^{-1}$
and $|G|:=\det(G)$, we have
\begin{equation}\label{LB}
  -\Delta_G
  = -|G|^{-1/2} \partial_i |G|^{1/2} G^{ij} \partial_j
  = - f^{-1} \partial_1 f^{-1} \partial_1
  - f^{-1} \partial_2 f \partial_2
  \,.
\end{equation}
Here the first equality
(in which the Einstein summation convention is assumed)
is a general formula for
the Laplace-Beltrami operator~$-\Delta_G$ expressed in local coordinates
in a Riemannian manifold equipped with a metric~$G$.
The second equality uses the special form~\eqref{metric},
for which $|G|=f^2$ and $G^{-1}=\diag(f^{-2},1)$.
Henceforth we assume that the Jacobian of~\eqref{Ass.basic}
is uniformly positive and bounded, \ie,
\begin{equation}\label{est.f}
  \fbox{\ $f, f^{-1} \in L^\infty(\Omega_0)$\,,}
\end{equation}
so that~$-\Delta_G$ is a uniformly elliptic operator.
Again, \eqref{est.f}~can be achieved for sufficiently small~$a$,
\cf~\eqref{Jacobi}.

\begin{remark}\label{Rem.Ass.basic}
The assumption~\eqref{Ass.basic} is not really essential.
Indeed, abandoning the geometrical interpretation of~$\Omega$
as a tubular neighbourhood embedded in~$\mathcal{A}$,
$(\Omega_0,G)$ with~\eqref{metric}
can be considered as an abstract Riemannian manifold
for which~\eqref{est.f} is the only important hypothesis.
The results of this paper extend automatically
to this more general situation.
\end{remark}

\subsubsection{The boundary conditions}\label{bound.con}

We denote $\partial_{i} \Omega_0=\Gamma_{i}^-\cup\Gamma_{i}^+$
the boundary in $x_{i}$ direction, $i\in \{1,2\}$, see Figure~\ref{ptstrip},
\begin{equation}\label{Gamma12}
\Gamma_1^{\pm}:=\{\pm l \}\times \Jd,
\qquad
\Gamma_2^{\pm}:=\Jj \times \{\pm a \}.
\end{equation}
Boundary conditions imposed respectively on
$\partial_{1} \Omega_0$ and $\partial_{2} \Omega_0$
are of different nature.
Having in mind the situation when~$\Gamma$ is a closed curve,
standard periodic boundary conditions are imposed on $\partial_{1} \Omega_0$, \ie,
\eq{\psi(-l,x_2)=\psi(l,x_2), \qquad
\partial_1 \psi(-l,x_2)=\partial_1 \psi(l,x_2)\label{PerBC},}
for a.e.\ $x_2 \in J_2$,
where~$\psi$ denotes any function from the domain of~$H$. We assume also the symmetry condition on the geometry
\begin{equation}\label{per.f}
  \fbox{\ $\forall(x_1,x_2) \in \Omega_0: \quad
  f(-l,x_2)=f(l,x_2),$}
\end{equation}
in order to have indeed periodic system in $x_1$ direction.

On the other hand, non-self-adjoint $\PT$-symmetric boundary conditions
are imposed on~$\partial_2\Omega_0$. A general form of~\mbox{$\PT$-symmetric}
boundary conditions was presented in~\cite{albeverio-2002-59}; further study and more general approach
to extensions can be found in \cite{Albeverio2005-38,Albeverio-2009-42}.
Denoting
\eq{
  \Psi:=
  \begin{pmatrix}
    \psi \\ \partial_2 \psi
  \end{pmatrix}
  ,
  \label{Psi} }
there are two types of the conditions, separated and connected.
\begin{subequations}\label{BC}
\renewcommand{\theequation}{\theparentequation$\,$\Roman{equation}}
\noindent
\begin{enumerate}
\item[\II.] separated:
\begin{equation}\label{SepBC}
  \begin{pmatrix}
    \pm\beta(x_1)+\ii \alpha(x_1) & 0 \\
    0 & 1
  \end{pmatrix}
  \Psi(x_1,\pm a)=0
\end{equation}
for a.e.\ $x_1 \in J_1$,
with $\alpha, \beta$ being real-valued functions.

\item[\I.] connected:
\begin{align}
\Psi(x_1,a) = B(x_1) \Psi(x_1,-a), \label{ConBC}
\end{align}
for a.e.\ $x_1 \in J_1$,
where the matrix $B$ has the form
\begin{align}
B(x_1):=\matice{\sqrt{1+b(x_1)c(x_1)}\,e^{\ii\phi(x_1)} & b(x_1) \\ c(x_1) &
 \sqrt{1+b(x_1)c(x_1)}\,e^{-\ii\phi(x_1)}} \nonumber
\end{align}
with $b,c,\phi$ being real-valued
functions satisfying $b > 0$, $c \geq -1/b$, $\phi\in[-\pi,\pi)$.
\end{enumerate}
\end{subequations}
We specify assumptions on smoothness, boundedness and periodicity
of the functions entering the boundary conditions later.
The index $\iota\in \{ {\II, \I }\}$ will be used throughout
the paper to distinguish between the two types of boundary conditions.

The boundary conditions (\ref{BC}$\iota$) are $\PT$-symmetric in following sense:
if a function~$\psi$ satisfies (\ref{BC}$\iota$),
then the function $\PT \psi$ satisfies (\ref{BC}$\iota$) as well.
Here and in the sequel
the symmetry operators~$\mathcal{P}$ and~$\mathcal{T}$
are defined as follows:
\begin{equation}\label{def.parity}
  (\mathcal{P}\psi)(x_1,x_2):=\psi(x_1,-x_2) \,,
  \qquad
  \mathcal{T}\psi:=\overline{\psi} \,.
\end{equation}

It is important to stress that
the $\PT$-symmetric boundary conditions~(\ref{BC}$\iota$)
do not automatically imply that the operator $H$ is $\PT$-symmetric,
unless additional assumption on the geometry of $\Omega_0$ is imposed.
The assumption, ensuring the $\PT$-symmetry of~$H$
(\cf~Proposition~\ref{PTpseudo} below), reads
\begin{equation}\label{sym.f}
  \fbox{\ $\forall(x_1,x_2) \in \Omega_0: \quad
  f(x_1,x_2)=f(x_1,-x_2)$.}
\end{equation}
In view of~\eqref{Jacobi},
a necessary condition to satisfy
the second requirement in~\eqref{sym.f}
is that the curve~$\Gamma$ is a geodesic, \ie~$\kappa=0$.

\subsubsection{The functional spaces}

The space in which we give a precise meaning of~$H$
is the Hilbert space $L^2(\Omega_0,G)$, \ie,
the class of all measurable functions~$\varphi,\psi$ on~$\Omega_0$
for which the norm $\|\cdot\|_G$
induced by the inner product
\begin{equation}\label{G.ip}
  (\varphi,\psi)_G
  := \int_{\Omega_0} \overline{\varphi(x)}\,\psi(x) \ |G(x)|^{1/2} \,\dd x
\end{equation}
is finite.
Assuming~\eqref{est.f}, the norm~$\|\cdot\|_G$ in $L^2(\Omega_0,G)$
is equivalent to the usual one $\|\cdot\|$ in $L^2(\Omega_0)$.
Moreover, the `energy space'
\begin{equation}\label{W12}
 W^{1,2}(\Omega_0,G) := \big\{
  \psi \in L^2(\Omega_0,G) \ \big| \
  |\nabla_G\psi|_G^2 := \overline{\partial_i\psi} G^{ij} \partial_j\psi
  \in L^2(\Omega_0,G)
  \big\}
\end{equation}
can be as a vector space identified with the usual
Sobolev space $W^{1,2}(\Omega_0)$.

However, this equivalence does not hold for~$W^{2,2}$-spaces,
unless one assumes extra regularity condition on~$f$:
\begin{equation}\label{reg.f}
 \fbox{\ $ \forall x_2 \in \Jd: \quad
  f(\cdot,x_2), \ f^{-1}(\cdot,x_2) \in W^{1,\infty}\big(\Jj\big)$.}
\end{equation}
Under this assumption,
which is actually equivalent to the Lipschitz continuity of $f, f^{-1}$
in the first argument (\cf~\cite[Chapt.~5.8.2.b., Thm.~4]{Evans1998}),
one can indeed identify the~$W^{2,2}$-Sobolev space
on the Riemannian manifold $(\Omega_0,G)$
(precisely defined, \eg, in~\cite[Sec.~2.2]{Hebey})
with the usual Sobolev space $W^{2,2}(\Omega_0)$.

\subsubsection{The schism: two definitions of the Hamiltonian}

Although the above equivalence of the~$W^{2,2}$-spaces under
the condition~\eqref{reg.f} is not explicitly used in this paper,
it is in fact hidden in our proof that the particle Hamiltonian
on $L^2(\Omega_0,G)$ naturally identified with
\begin{subequations}\label{H.operator}
\begin{align}
H_{\iota}\psi&:=-\Delta_G \psi, \label{H.action} \\
\psi\in\Dom(H_{\iota})
&:= \big\{ \psi \in W^{2,2}(\Omega_0)\
\big| \ \psi \ {\rm satisfies \ } \eqref{PerBC} \
{\rm and\ } (\ref{BC}\iota) \big\}. \label{DomH}
\end{align}
\end{subequations}
is well defined (\cf~Theorem~\ref{OpForm}).
As mentioned in Section~\ref{bound.con},
we use the notation~$H_{\iota}$, with $\iota\in\{\rm \II,\I \}$,
to distinguish between separated~\eqref{SepBC}
and connected~\eqref{ConBC} boundary conditions.

To avoid the additional assumption~\eqref{reg.f},
one can always interpret~\eqref{LB} in the weak sense of quadratic forms,
which gives rise to an alternative Hamiltonian~$\tilde{H}_\iota$
(\cf~Corollary~\ref{Corol.m}).
This is the content of the following section,
where we also show that $H_{\iota}=\tilde{H}_\iota$
provided that~\eqref{per.f}, \eqref{reg.f}, and some analogous hypotheses
about the boundary-coupling functions hold.

\section{General properties}\label{Sec.form}
The main goal of this section is to show that
the Hamiltonian~$H_\iota$ introduced in~\eqref{H.operator}
is a well defined operator, in particular that it is closed.
This will be done by proving that $H_{\iota}=\tilde{H}_\iota$,
where~$\tilde{H}_\iota$ is the alternative operator defined
through a closed quadratic form.
Finally, we establish some general spectral properties
of the Hamiltonians.

\subsection{The Hamiltonian defined via quadratic form}
Taking the sesquilinear form $(\varphi,H_\iota\psi)_G$
with $\varphi, \psi \in \Dom(H_{\iota})$
and integrating by parts, one arrives to a sesquilinear form,
which is well defined for a wider class of functions $\varphi, \psi$,
not necessarily possessing second (weak) derivatives.
The~function $f$ is assumed to satisfy \eqref{est.f} and \eqref{per.f},
however the extra regularity condition~\eqref{reg.f} is \emph{not} required.

More precisely, exclusively under assumption \eqref{est.f} for a moment,
we define the sesquilinear form
\begin{align*}
h_{\iota}(\varphi, \psi)&:=h^1(\varphi, \psi)+h_\iota^2(\varphi, \psi),
\\
\varphi, \psi \in \Dom(h_{\iota})
&:= W^{1,2}_{\rm per}(\Omega_0)
\equiv \left\{ \psi\in W^{1,2}(\Omega_0)\ \big| \ \psi(-l,x_2)=\psi(l,x_2) \right\},
\end{align*}
where, for any $\varphi, \psi \in \Dom(h_{\iota})$,
\begin{eqnarray*}
h^1(\varphi,\psi)&:=&\big(f^{-1}\partial_1 \varphi,f^{-1} \partial_1 \psi\big)_G +
                            \big(\partial_2 \varphi,\partial_2 \psi\big)_G \,,  \\
h_{\II}^2(\varphi,\psi)&:=&  \big( \varphi,( \beta+\ii \alpha) \psi\big)_G^{\Gamma_2^+}
+ \big( \varphi,(\beta-\ii \alpha) \psi\big)_G^{\Gamma_2^-},        \\
h_{\I}^2(\varphi,\psi)&:=& \big( \varphi, B_{12}^{-1} \P \psi\big)_G^{\Gamma_2^+}
+ \big( \varphi,B_{12}^{-1}
\P \psi\big)_G^{\Gamma_2^-}  \\
&&  - \big( \varphi, B_{22}B_{12}^{-1} \psi\big)_G^{\Gamma_2^+}
- \big( \varphi, B_{11}B_{12}^{-1} \psi\big)_G^{\Gamma_2^-} .
\end{eqnarray*}
Here~$B_{ij}$ denotes the elements of the matrix~$B$ defined in~\eqref{BC},
the operator~$\P$ is introduced in~\eqref{def.parity} and
\begin{equation*}
(\varphi,\psi)_G^{\Gamma_2^{\pm}}
:= \int_{-l}^{l}{\overline{\varphi(x_1,\pm a)}\,\psi(x_1,\pm a) \, f(x_1,\pm a) \, \dd x_1 }.
\end{equation*}
All the boundary terms
should be understood in sense of traces~\cite{Adams1975}.
\begin{lemma} \label{h1.prop}
Let $f$ satisfy \eqref{est.f}. The forms $h_{\iota}, h^1$ are densely defined. $h^1$ is a symmetric, positive, closed form (associated to the self-adjoint Laplace-Beltrami operator in $L^2(\Omega_0,G)$ with periodic boundary conditions on $\partial_1 \Omega_0$ and Neumann boundary conditions on $\partial_2 \Omega_0$).
\end{lemma}
\begin{proof}
The density of the domains is obvious, properties of $h^1$ are well known,
see the detailed discussion on a similar problem in~\cite[Sect.~7.2]{Davies1995}.
\end{proof}
Although the forms $h_{\iota}$ are not symmetric, we show that $h^2_{\iota}$ can be understood
as small perturbations of $h^1$.
\begin{lemma}\label{RelBound}
Let $b,1/b,c,\alpha,\beta \in L^{\infty}(\Jj)$ and let $f$ satisfy \eqref{est.f}. Then $h_{\iota}^2$ are relatively bounded
with respect to $h^1$ with
\begin{equation}\label{rel.bound}
\begin{aligned}
|h^2_{\iota}[\psi]|
&\leq \varepsilon \, h^1[\psi]+ \varepsilon^{-1} C \|\psi\|_G^2,
%
%
\end{aligned}
\end{equation}
for all $\psi \in W^{1,2}_{\rm per}(\Omega_0)$ and any positive number $\varepsilon$.
The constant~$C$ depends on the function $f$, dimensions $a,l$, and boundary-coupling functions $\alpha,\beta$ or $b,c,\phi.$
\end{lemma}
\begin{proof}
The proof is based on the estimate
\eq{\int_{-l}^{l}\left|\psi(x_1,\pm a)\right|^2\dd x_1
\leq \epsilon \, \| \nabla \psi \|^2 + \epsilon^{-1} \tilde{C} \, \|\psi\|^2,  \label{W12.est}}
where~$\epsilon$ is an arbitrary positive constant
and~$\tilde{C}$ is a positive constant depending only on~$a$ and~$l$.
We give the proof for $h_{\II}^2$ only because the other case is analogous.
The~assumptions on $\alpha,\beta$ and property \eqref{est.f} allow us to estimate
the functions $|\alpha|, |\beta|$ and $f$ by their $L^{\infty}$-norms.
Consequent application of \eqref{W12.est} therefore yields
$$
  \big|h^2_{\II}[\psi]\big|
  \leq \epsilon \,   \|f\|_{L^{\infty}(\Omega_0)}  \| \nabla \psi \|^2
  + \epsilon^{-1} 2 \, \tilde{C} \,
  \big(\|\alpha\|_{L^{\infty}(\Jj)}+\|\beta\|_{L^{\infty}(\Jj)}\big) \,
  \|f\|_{L^{\infty}(\Omega_0)} \, \|\psi\|^2.
$$
In order to replace the term $\| \nabla \psi \|^2$ by $h^1[\psi]$,
the regularity assumption on geometry \eqref{est.f} is used. Once we consider the equivalence of the norms $\|\cdot\|$ and $\|\cdot\|_G$ and the arbitrariness of $\epsilon$, we obtain the estimate \eqref{rel.bound}.
\end{proof}
\begin{corollary}\label{Corol.m}
Let $b,1/b,c,\alpha,\beta \in L^{\infty}(\Jj)$ and let $f$ satisfy \eqref{est.f}.
Then there exist the unique \mbox{m-sectorial}
operators $\tilde{H}_{\iota}$ in $L^2(\Omega_0,G)$ such that
\begin{equation}
h_{\iota}(\varphi,\psi) =: (\varphi,\tilde{H}_{\iota}\psi)_G
\end{equation}
for all $\psi \in \Dom(\tilde{H}_{\iota})$ and $\varphi \in \Dom(h_{\iota})$, where
\begin{eqnarray} \label{DomTilde}
\Dom(\tilde{H}_{\iota}):=\big\{ \psi\in W^{1,2}_{\rm per}(\Omega_0) & \big|& \
\exists F \in L^2(\Omega_0,G), \ \forall \varphi\in W^{1,2}_{\rm per}(\Omega_0), \nonumber \\
 && h_{\iota}(\varphi,\psi)=(\varphi,F)_{G}          \big\}.
\end{eqnarray}
\end{corollary}
\begin{proof}
With regard to Lemmata \ref{h1.prop}, \ref{RelBound},
and the perturbation result \cite[Thm.~VI.3.4]{Kato},
the statement follows by
the first representation theorem \cite[Thm.~VI.2.1]{Kato}.
\end{proof}

\subsection{The equivalence of the two definitions}

Under stronger assumptions on smoothness of functions appearing in boundary conditions (\ref{BC}$\iota$)
and on the function $f$ entering the metric tensor $G$, we show that operators $\tilde{H}_{\iota}$ associated to
the forms $h_{\iota}$ are equal to the Hamiltonians $H_{\iota}$ defined in \eqref{H.operator}.
To prove this, we need the following lemma.
Let us introduce a space of Lipschitz continuous functions over~$[-l,l]$
satisfying periodic boundary conditions:
$$
W^{1,\infty}_{\rm per}\big(\Jj\big)
:= \big\{ \psi \in W^{1,\infty}\big(\Jj\big) \, \big|\ \psi(-l)=\psi(l) \big\}.
$$

\begin{lemma}\label{Dom.Hcs}
Let $\alpha, \beta, b, 1/b, c,\phi \in W^{1,\infty}_{\rm per}\big(\Jj\big)$ and
let $f$ satisfy \eqref{est.f}, \eqref{per.f}, and \eqref{reg.f}. Then for every $F\in L^2(\Omega_0,G)$,
a solution $\psi$ to the problem
\begin{equation}\label{genEq}
\forall \varphi \in W^{1,2}_{\rm per}(\Omega_0) \,, \qquad
h_{\iota}(\varphi,\psi)=(\varphi,F)_{G} \,,
\end{equation}
belongs to $\Dom(H_{\iota})$ introduced in~\eqref{DomH}.
\end{lemma}
\begin{proof}
We prove the separated boundary conditions case only, the connected case is analogous.
For each $\psi\in W^{1,2}_{\rm per}(\Omega_0)$
We introduce a difference quotient
\eq{
\delta\psi(x_1,x_2):=\frac{\psi_{\delta}(x_1,x_2)-\psi(x_1,x_2)}{\delta},}
where $\psi_{\delta}(x_1,x_2):=\psi(x_1+\delta,x_2)$
and~$\delta$ is a small real number. The shifted value $\psi_{\delta}(x_1,x_2)$ is well defined for every
$x_1 \in \Jj$ and $\delta\in \R$ by extending $\psi$ periodically to $\R$. We use periodic extensions of other functions in $x_1$ direction throughout the whole proof without further specific comments.
The estimate
\eq{\|\delta\psi\| \leq \|\psi\|_{W^{1,2}(\Omega_0)}   \label{Diference}}
is valid for $\delta$ small enough \cite[Sec.~5.8.2., Thm.~3]{Evans1998}.

We express the difference of identities \eqref{genEq} for $\psi$ and $\psi_{\delta}$,
whence we get for every $\varphi\in W^{1,2}_{\rm per}(\Omega_0)$
\begin{eqnarray}\label{Id}
\big(\partial_1\varphi, (\delta f^{-1}) \partial_1 \psi\big)
+\big(\partial_1\varphi, f^{-1}_{\delta}\partial_1 (\delta \psi)\big)
+\big(\partial_2\varphi,(\delta f) \partial_2 \psi\big) \nonumber \\
+ \big(\partial_2\varphi, f_{\delta} \partial_2 (\delta \psi) \big)
+ \big(\varphi, \delta (f(\beta+\ii \alpha)) \psi_{\delta}  \big)^{\Gamma_2^+}
+ \big(\varphi, f (\beta+\ii \alpha) (\delta \psi)   \big)^{\Gamma_2^+} \nonumber \\
+ \big(\varphi,\delta(f (\beta-\ii \alpha))\psi_{\delta}  \big)^{\Gamma_2^-}
+ \big(\varphi, f(\beta-\ii \alpha) (\delta \psi)   \big)^{\Gamma_2^-} \nonumber \\
= \big(\varphi, (\delta f) F_{\delta}  \big) + \big(\varphi, f (\delta F) \big),
\end{eqnarray}
where $(\cdot,\cdot)$ is the inner product in $L^2(\Omega_0)$ and
\begin{equation}\label{b.scal.prod}
(\varphi,\psi)^{\Gamma_2^{\pm}}
:= \int_{-l}^{l}{\overline{\varphi(x_1,\pm a)}\,\psi(x_1,\pm a) \, \dd x_1 }.
\end{equation}
We insert $\varphi=\delta \psi$ into equation~\eqref{Id}
and apply the `integration-by-parts' formula \cite[Sec.~5.8.2]{Evans1998}
for difference quotients, \ie,
$(\varphi,\delta F \big)=-\big((-\delta) \varphi, F \big)$, in order to avoid the difference quotient of the arbitrary (\eg~possibly non-continuous) function $F\in L^2(\Omega_0,G)$.
Using the embedding of $W^{1,2}(\Omega_0)$ in $L^2(\partial\Omega_0)$,
the regularity assumptions on $\alpha$, $\beta$ and $f$,
the Schwarz and Cauchy inequalities,
and the estimate~\eqref{Diference}, we obtain
\eq{\|\delta \psi\|_{W^{1,2}(\Omega_0)}\leq C,    }
where $C$ is a constant independent of $\delta$. By standard arguments \cite[D.4]{Evans1998},
this estimate yields that $\partial_1 \psi \in W^{1,2}(\Omega_0)$.

At the same time,
standard elliptic regularity theory \cite[Thm.~8.8]{Gilbarg-Trudinger}
implies that the solution~$\psi$ to~\eqref{genEq}
belongs to $W^{2,2}_{\rm loc}(\Omega_0)$.
Thus~$\psi$ satisfies the equation
\begin{equation}\label{LB.a.e}
-\Delta_{G}\psi=F
\end{equation}
a.e.~in $\Omega_0$. If we express $\partial_{2}^2 \psi$ from \eqref{LB.a.e},
we obtain that
\mbox{$\partial_{2}^2 \psi \in L^2(\Omega_0)$}.

It remains to check boundary conditions of $\Dom(H_{\II})$.
Once the $W^{2,2}$-regularity of the solution~$\psi$ is established,
this can be done by using integration by parts
in the identity~\eqref{genEq}
and considering the arbitrariness of~$\varphi$, see \cite[Lemma 3.2]{borisov-2007}
for the more detailed discussion in an analogous situation.
\end{proof}
Let us write $H_{\II}(\alpha,\beta)$ and $H_{\I}(b,c,\phi)$ if we want to stress the dependence of the Hamiltonians on functions $\alpha, \beta$ and $b,c,\phi$ entering the boundary conditions.
\begin{theorem}\label{OpForm}
Let $\alpha, \beta, b, 1/b, c,\phi \in W^{1,\infty}_{\rm per}\big(\Jj\big)$ and let $f$ satisfy \eqref{est.f}, \eqref{per.f}, and \eqref{reg.f}. Then
\begin{enumerate}
\item $\tilde{H}_{\iota}=H_{\iota}$,
\item $H_{\iota}$ are m-sectorial operators,
\item the adjoint operators $H_{\iota}^*$ can be found as
        \begin{equation*}
        H_{\II}^*(\alpha,\beta)=H_{\II}(-\alpha,\beta),
        \qquad
        H_{\I}^*(b,c,\phi)=H_{\I}(b,c,-\phi),
        \end{equation*}
\item the resolvents of $H_{\iota}$ are compact.
\end{enumerate}

\end{theorem}
\begin{proof}
\emph{Ad~1.}
It is easy to verify, by integration by parts,
that if $\psi \in \Dom(H_{\iota})$ then
$\psi \in \Dom(\tilde{H}_{\iota})$;
in fact, the function~$F$ from~\eqref{DomTilde}
satisfies $F=-\Delta_{G}\psi$ in the distributional sense.
Thus $H_{\iota} \subset \tilde{H}_{\iota}$. The more non-trivial inclusion
$\tilde{H}_{\iota} \subset H_{\iota}$ follows from Lemma~\ref{Dom.Hcs}.
Once the equality of the operators is established,
the other properties readily follow from the corresponding
properties for~$\tilde{H}_\iota$.

\emph{Ad~2.}
$\tilde{H}_\iota$~is m-sectorial by Corollary~\ref{Corol.m}.

\emph{Ad~3.}
By \cite[Thm. VI.2.5]{Kato}, the adjoint operator~$\tilde{H}_\iota^*$
is associated to the adjoint form
$h_{\iota}^*(\varphi,\psi):=\overline{h_{\iota}(\psi,\varphi)}$,
which establishes the required identities for~$\tilde{H}_\iota$.

\emph{Ad~4.}
The compactness of the resolvents for $\tilde{H}_{\iota}$ is provided by the perturbation result
\cite[Thm. VI.3.4]{Kato} and Lemmata~\ref{h1.prop}, \ref{RelBound}.
\end{proof}

\subsection{Spectral consequences}

Since the Hamiltonians~$H_{\iota}$ are m-sectorial by Theorem~\ref{OpForm},
the spectrum (as a subset of the numerical range)
is contained in a sector of the complex plane,
\ie, there exists a vertex~$\gamma\in\R$
and a semi-angle $\theta\in[0,\pi/2)$ such that
\begin{equation*}
\sigma(H_{\iota}) \subset \big\{\zeta \in \C
\ \big| \
|\arg(\zeta-\gamma)|\leq \theta \big\}.
\end{equation*}
Furthermore, since the resolvents of~$H_{\iota}$ are compact,
the spectra of~$H_{\iota}$ are purely discrete,
as it is reasonable to expect for the Laplacian
defined on a bounded manifold.

Under the additional assumptions on the geometry
of the model~\eqref{sym.f}, one can show that~$H_{\iota}$
are \mbox{$\PT$-symmetric}.
\begin{proposition}\label{PTpseudo}
Let $\alpha, \beta, b, 1/b, c,\phi \in W^{1,\infty}_{\rm per}\big(\Jj\big)$ and let $f$ satisfy \eqref{est.f}, \eqref{sym.f},
and~\eqref{reg.f}.
Then Hamiltonians $H_{\iota}$ are
\begin{enumerate}
\item $\PT$-symmetric, \ie,
$ (\PT) H_{\iota} \subset H_{\iota} (\PT)$,
\item $\P$-pseudo-Hermitian, \ie,
$ H_{\iota}=\P H_{\iota}^* \P, $
\item $\T$-self-adjoint, \ie,
$ H_{\iota}=\T H_{\iota}^* \T,$
\end{enumerate}
where the operators~$\P$ and~$\T$ are defined in~\eqref{def.parity}.
\end{proposition}
\begin{proof}
Note that the $\PT$-symmetry relation means that whenever $\psi\in\Dom(H_\iota)$,
$\PT\psi$ also belongs to~$\Dom(H_\iota)$
and $\PT H_{\iota}\psi = H_{\iota}\PT\psi$.
This can be verified directly
using the definition of~$H_{\iota}$ via~\eqref{H.operator}.
The proofs of the remaining statements are based
on the explicit knowledge of the adjoint operators, Theorem~\ref{OpForm}.3.
\end{proof}
\begin{corollary}\label{Complex.pairs}
Under the hypotheses of Proposition~\ref{PTpseudo},
the spectra of~$H_{\iota}$ are invariant under complex conjugation, \ie,
$$
  \forall \lambda\in\C \,, \qquad
  \lambda \in \sigma(H_{\iota})
  \Longleftrightarrow \overline{\lambda} \in \sigma(H_{\iota})
  \,.
$$
\end{corollary}
\begin{proof}
Recall that the spectrum of~$H_\iota$ is purely discrete
due to Theorem~\ref{OpForm}.4.
With regard to $\PT$-symmetry, it is easy to check that if $\psi$
is the eigenfunction corresponding to the eigenvalue $\lambda$, then $\PT\psi$
is the eigenfunction corresponding to the eigenvalue $\overline{\lambda}$.
\end{proof}

\section{Solvable models: constantly curved manifolds}\label{Sec.constant}

In order to examine basic effects of curvature on the spectrum
of the Hamiltonians we investigate solvable models now.
We restrict ourselves to the spectral problem
in constantly curved manifolds and subjected to constant
interactions on the boundary,
\ie, the functions $K, \alpha, \beta, b, c, \phi$
are assumed to be constant.
Moreover, we assume that~$\Gamma$ is a geodesic, \ie~$\kappa=0$, to have \eqref{sym.f}.

\subsection{Preliminaries}

To emphasize the dependence of the Hamiltonians~$H_\iota$ on the curvature~$K$,
we use the notation~$H_{\iota(K)}$ in this section.
One can easily derive the scaling properties of eigenvalues
for constant $K\neq 0$:
$$
\lambda_{\iota}(K,a,l)
\,=\, |K|\,\lambda_{\iota}\big(\pm 1,\sqrt{|K|}\,a,\sqrt{|K|}\,l\big).
$$
Hence, the decisive factor for qualitative properties of the spectrum
is the sign of~$K$, while the specific value of curvature is not essential.
Hereafter we restrict ourselves to
\begin{equation}\label{K-const}
  K\in\{-1,0,1 \} \,.
\end{equation}
Possible realizations of the ambient manifolds~$\mathcal{A}$ corresponding
to these three cases are pseudosphere, cylinder, and sphere, respectively,
see Figure \ref{realizations}.

\begin{figure}[ht]
\begin{center}
\subfloat[$K=-1$, pseudosphere]{\includegraphics[width=0.3\textwidth, height=0.4\textwidth]{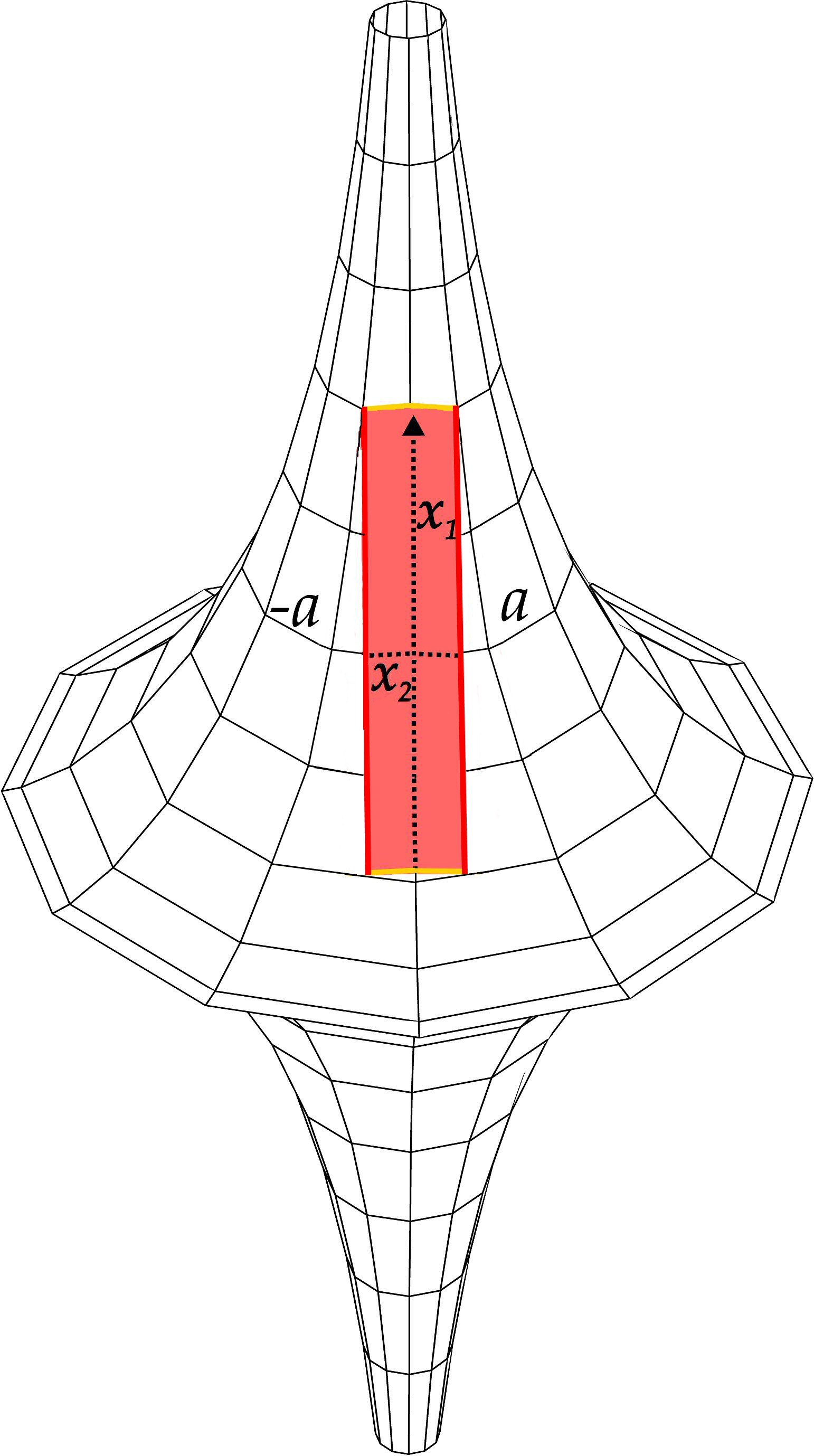}}\hspace{0.2cm}
\subfloat[$K=0$, cylinder]{\includegraphics[height=0.4\textwidth, width=0.27\textwidth]{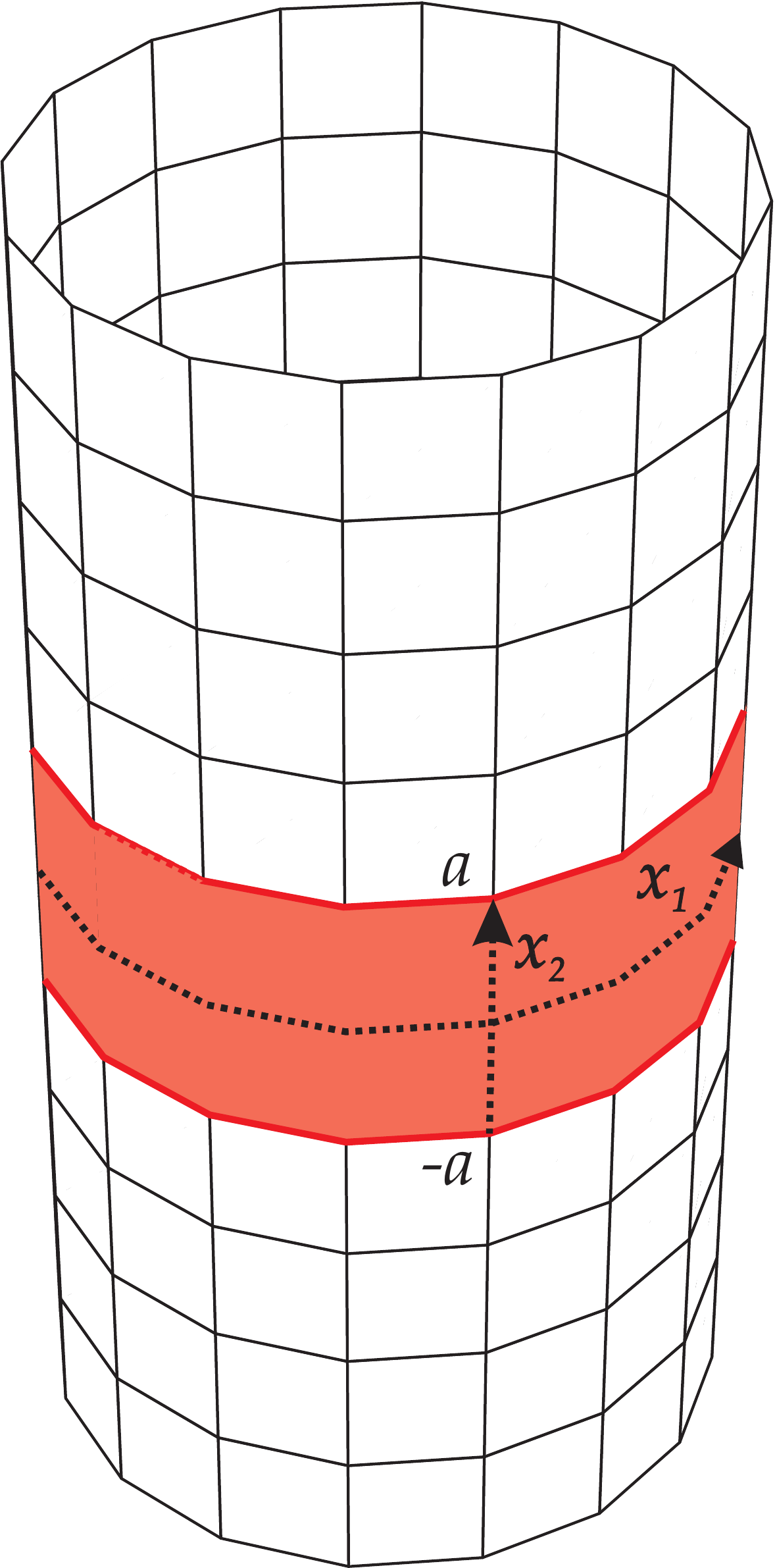}} \hspace{0.2cm}
\subfloat[$K=1$, sphere]{\includegraphics[width =0.37\textwidth]{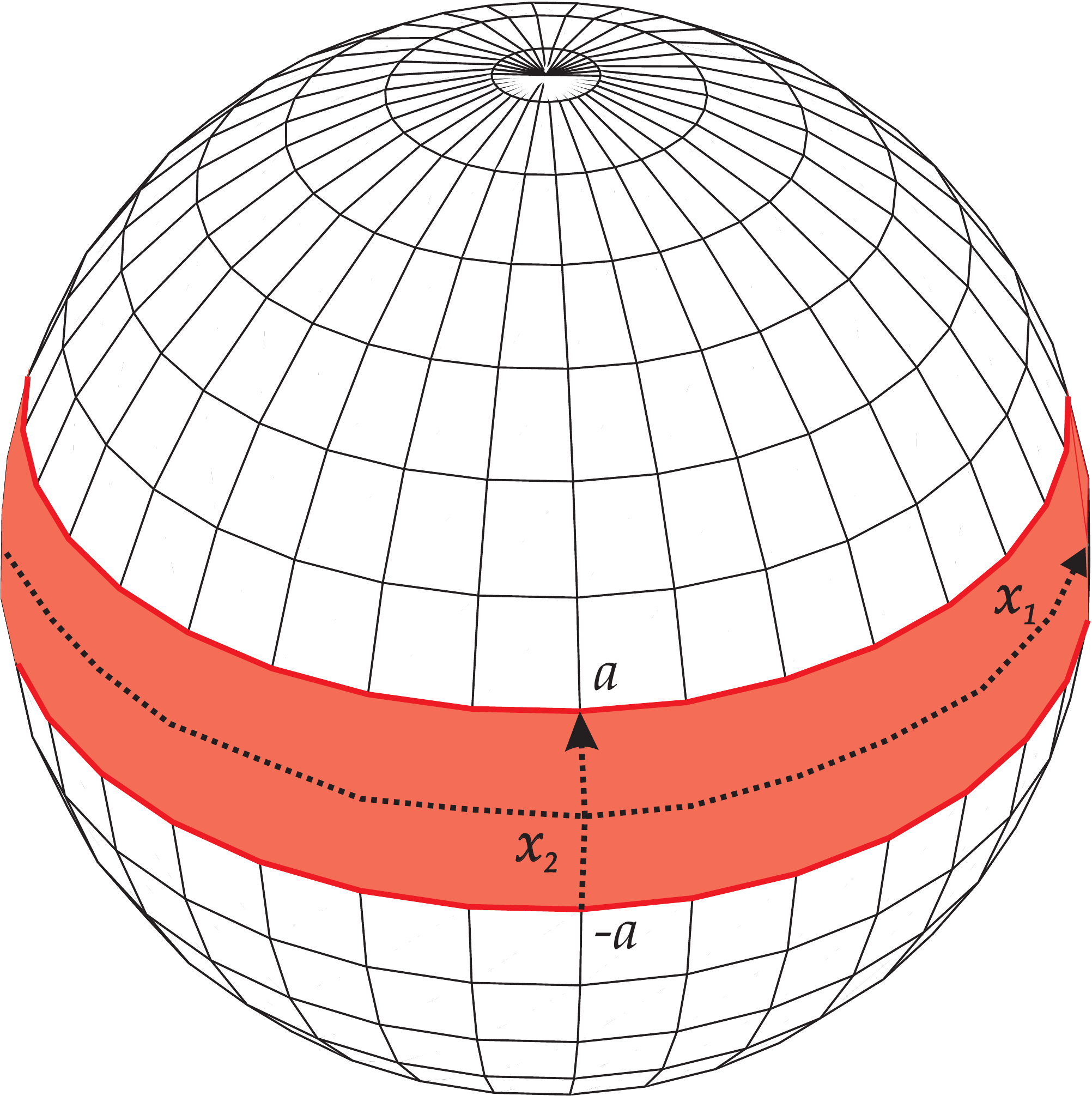}}
\caption {Realizations of the constantly curved manifolds.}
\label{realizations}
\end{center}
\end{figure}

\begin{remark}
The pseudosphere should be considered as a useful
realization of~$\mathcal{A}$ with $K=-1$ only locally,
since no complete surface of constant negative
curvature can be globally embedded in~$\R^3$
(this is reflected by the singular equator in Figure~\ref{realizations}.(a)).
However, since~$\Omega$ is a precompact subset of~$\mathcal{A}$,
the incompleteness of the pseudosphere is not a real obstacle here.
\end{remark}

Moreover, hereafter we put $l=\pi$,
so that the length of the strip is $2\pi$.
This provides an instructive visualization of~$\Omega$
as a tubular neighbourhood of a geodesic circle
on the cylinder and the sphere, see Figure \ref{realizations}.

For $\kappa=0$ and constant curvatures~\eqref{K-const},
the Jacobi equation~\eqref{Jacobi} admits the explicit solutions
\begin{equation}\label{Jacobi-const}
  f_{(K)}(x_1,x_2) =
  \begin{cases}
    \cosh x_2 & \mbox{if} \quad K=-1 \,, \\
    1 & \mbox{if} \quad K=0 \,, \\
    \cos x_2 & \mbox{if} \quad K=1 \,.
  \end{cases}
\end{equation}
It follows that the assumption~\eqref{est.f}
is satisfied for any positive~$a$ if $K=-1,0$,
while one has to restrict to $a < \pi/2$ if $K=1$.
The latter is also sufficient to satisfy~\eqref{Ass.basic} for the sphere.
There is no restriction on~$a$ to have~\eqref{Ass.basic}
if~$\Gamma$ is the geodesic circle on the cylinder.
In any case (including the pseudosphere),
\eqref{Ass.basic}~can be always satisfied for sufficiently small~$a$.
The other hypotheses, \ie~\eqref{per.f}, \eqref{sym.f}, and~\eqref{reg.f}, clearly hold
regardless of the curvature sign.
\begin{remark}
In view of Remark~\ref{Rem.Ass.basic},
$a < \pi/2$ for $K=1$ is the only essential restriction
in the constant-curvature case~\eqref{Jacobi-const}.
\end{remark}

Explicit structures of the Hamiltonians~$H_{\iota(K)}$
introduced in~\eqref{H.operator} readily follow
from~\eqref{LB} by using~\eqref{Jacobi-const}:
\begin{equation}\label{DefOpKP}
  H_{\iota(K)} =
  \begin{cases}
    \displaystyle
    -\frac{1}{\cosh^2x_2}\partial_1^2-\partial_2^2-\tanh x_2 \partial_2
    & \mbox{if} \quad K=-1 \,,
    \medskip \\
    \displaystyle
    -\partial_1^2-\partial_2^2
    & \mbox{if} \quad K=0 \,,
    \medskip \\
    \displaystyle
    -\frac{1}{\cos^2x_2}\partial_1^2-\partial_2^2+\tan x_2 \partial_2
    & \mbox{if} \quad K=1 \,,
  \end{cases}
\end{equation}
on $\Dom(H_{\iota(K)})$.

\subsection{Partial wave decomposition}

Since both the coefficients of~$H_{\iota(K)}$
and the boundary conditions are independent
of the first variable~$x_1$,
we can decompose the Hamiltonians
into a direct sum of transverse one-dimensional operators.
The decomposition is based on the following lemma.
\begin{lemma}
\begin{equation}\label{eimx1.dec.1}
\forall \Psi \in L^2(\Omega_0,G), \qquad
\Psi(x_1,x_2)=\sum_{m\in \Z} \psi_m(x_2) \phi_m(x_1)
\quad {\rm in \quad } L^2(\Omega_0,G),
\end{equation}
where
\begin{equation}\label{eimx1.dec.2}
\phi_m(x_1):=\frac{1}{\sqrt{2\pi}}  e^{\ii m x_1}, \qquad
\ \psi_m(x_2):= \big(\phi_m, \Psi(\cdot,x_2) \big)_{L^2(\Jj)}.
\end{equation}
\end{lemma}
\begin{proof}
We may restrict the proof to $L^2(\Omega_0)$ only
because the norms $\|\cdot\|$ and $\|\cdot\|_{G}$
are equivalent due to~\eqref{est.f}. Let us also stress that $G$ is independent of $x_1$ and
$\big\{\phi_m\big\}_{m\in\Z}$ forms
an orthonormal basis of $L^2(\Jj)$. Hence
\begin{equation}\label{eimx1.dec.3}
\left\| \sum_{m\in\Z} \psi_m(x_2)\phi_m  \right\|_{L^2(\Jj)} = \|\Psi(\cdot,x_2)\|_{L^2(\Jj)} \in L^2(\Jd).
\end{equation}
The decomposition in $L^2(\Omega_0)$ can be then justified
by using the dominated convergence theorem.
\end{proof}

Writing $\Psi(x_1,x_2)=\sum_{m\in \Z}{\phi_m(x_1)\psi_m(x_2)}$
in the expression $H_{\iota(K)}\Psi$
and formally interchanging the summation
and differentiation,
we (formally) arrive at the decomposition:
\begin{equation}\label{H.dec}
H_{\iota(K)} = \bigoplus_{m\in \Z} H_{\iota(K)}^m B^m
\end{equation}
with
\begin{equation*}
  H_{\iota(K)}^m :=
  \begin{cases}
    \displaystyle
    -\partial_2^2-\tanh x_2\partial_2 +\frac{m^2}{\cosh^2x_2}
    & \mbox{if} \quad K=-1 \,, \\
    -\partial_2^2+m^2
    & \mbox{if} \quad K=0 \,, \\
    \displaystyle
    -\partial_2^2+\tan x_2\partial_2 +\frac{m^2}{\cos^2x_2}
    & \mbox{if} \quad K=1 \,,
  \end{cases}
\end{equation*}
where $B^m$ are bounded rank-one operators defined by
\begin{equation}\label{Bm}
  (B^m \Psi)(x_1,x_2)
  := \big(\phi_m,\psi(\cdot,x_2)\big)_{L^2(\Jj)}\, \phi_m(x_1) \,.
\end{equation}
The operators~$H_{\iota(K)}^m$ act in $L^2(\Jd,\dd \nu_{(K)})$ spaces
with the measure
\begin{equation}
  \dd \nu_{(K)}(x_2)
  :=
  \begin{cases}
    \cosh x_2 \, \dd x_2
    & \mbox{if} \quad K=-1 \,,
    \\
    \dd x_2
    & \mbox{if} \quad K=0 \,,
    \\
    \cos x_2 \, \dd x_2
    & \mbox{if} \quad K=1 \,.
  \end{cases}
\end{equation}
The domains of $H_{\iota(K)}^m$ are given by
\begin{equation}
\Dom (H_{\iota(K)}^m):= \big\{ \psi \in W^{2,2}(\Jd) \ \big| \
\psi \ {\rm satisfies \ } (\ref{BC}\iota) \big\},
\end{equation}
with obvious modification of the $\PT$-symmetric boundary conditions $(\ref{BC}\iota)$ to~the one-dimensional situation.

To justify the decomposition~\eqref{H.dec} in a resolvent sense,
we need the following technical lemma specifying
the numerical range of~$H_{\iota(K)}^m$.
\begin{lemma}\label{Num.range}
Let $\Xi_{\iota(K)}^m $ denote the numerical range of $H_{\iota(K)}^m$.
Then for every $m\neq 0$ there exist real constants $c_0, c_1$ independent of $m$ such that
\begin{equation}
\Xi_{\iota(K)}^m  \subset \left\{
z \in \C \ \big| \
 \Re z \geq c_0 + m^2, \ |\Im z| \leq c_1 \sqrt{\Re z + |c_0| - m^2 }
\right\}.
\end{equation}
\end{lemma}
\begin{proof}
We give the proof for $H_{\II(+1)}^m$ only,
the other cases are analogous.
We abbreviate $(\cdot,\cdot)_+ :=(\cdot,\cdot)_{L^2(\Jd,\dd \nu_{(+1)})}$ and define
\begin{equation*}
v^m(x_2):= \frac{m^2}{\cos^{2} x_2} \,, \qquad
h[\psi]:=\big( \psi,H_{\II(+1)}^m \psi \big)_+ \,,
\end{equation*}
for every $\psi \in \Dom (H_{\II(+1)}^m )$.
Integration by parts yields the following expressions
for real and imaginary parts of~$h[\psi]$:
\begin{eqnarray*}
\Re h[\psi] &=& \|\psi'\|_+^2 + (\psi, v^m \psi)_+ + \beta \cos a\left(|\psi(a)|^2+|\psi(-a)|^2 \right),
 \\
\Im h[\psi] &=& \alpha \cos a \big( |\psi(a)|^2-|\psi(-a)|^2 \big), 
\end{eqnarray*}
for every $\psi \in \Dom (H_{\II(+1)}^m )$.
The estimates of $\Re h[\psi]$ and $\Im h[\psi]$ can be easily obtained taking into account the equivalence of the norm $\|\cdot\|_{L^2(\Jd)}$ with $\|\cdot\|_+$ and using the one-dimensional version of the estimate \eqref{W12.est}.
\end{proof}

Now we are in a position to establish the main result of this subsection.
\begin{proposition}\label{decomposition}
$\displaystyle D:= \bigcap_{m\in \Z} \varrho\left(H_{\iota(K)}^m\right)$ is non-empty and
$D \subset \varrho\left(H_{\iota(K)}\right)$.
For every $z \in  D$,
\begin{equation}\label{resolvent.dec}
(H_{\iota(K)}-z)^{-1}
= \bigoplus_{m\in \Z} \big(H_{\iota(K)}^m-z\big)^{-1}B^m,
\end{equation}
where $(H_{\iota(K)}^m-z)^{-1}$ abbreviates
$1 \otimes (H_{\iota(K)}^m-z)^{-1}$ acting
on $L^2(\Jj)\otimes L^2(\Jd,\dd \nu_{(K)})$ and $B^m$ are defined in \eqref{Bm}.
\end{proposition}
\begin{proof}
We give a proof for $H_{\II(+1)}^m$ only, the remaining cases are analogous.
Take $z \in D$,
for every
$\Psi \in L^2(\Omega,G)$ and $m\in \Z$, we define
\begin{equation}
U_m(x_2):=\big(H_{\II(+1)}^m-z \big)^{-1}\psi_m(x_2),
\end{equation}
where $\psi_m$ was introduced in \eqref{eimx1.dec.2}. It is clear that $U_m \in L^2(\Jd,\dd \nu_{(+1)})$.
With regard to Lemma \ref{Num.range}, take $m_0 \in \Z$ such that for every $m > m_0,$ $z \notin \Xi_{\iota(K)}^m$. Using \cite[Thm. V.2.3]{Kato} together with Lemma \ref{Num.range},
we get for $m>m_0$
\begin{equation}\label{Um.est.1}
\|U_m\|_{L^2(\Jd)} \leq C_1 \frac{\|\psi_m\|_{L^2(\Jd)}}{m^2+1},
\end{equation}
where $C_1$ is a constant independent of $m$,
nonetheless depending on $z, |\alpha|,|\beta|,$ and $a$. Let us remark that since $z \in D$, $\|U_m\|_{L^2(\Jd)}$ are bounded for finitely many $m$ smaller than $m_0$.
From the identity
\begin{eqnarray*}
\|U_m'\|_+^2+(-\ii \alpha + \beta)\cos a |U_m(a)|^2
+ (-\ii \alpha - \beta)\cos a |U_m(-a)|^2 \\
+ (v^m U_m, U_m)_+ - \overline{z} \|U_m\|_+^2 = (\psi_m, U_m)_+ \,,
\end{eqnarray*}
with $v^m$ and $(\cdot,\cdot)_+$ defined in Lemma \ref{Num.range}, we obtain the estimate for the norm of $U_m'$ for $m>m_0$,
\begin{equation}\label{Um.est.2}
\|U_m'\|_{L^2(\Jd)} \leq C_1 \frac{\|\psi_m\|_{L^2(\Jd)}}{\sqrt{m^2+1}}
\end{equation}
Again, for finitely many $m\leq m_0$, $\|U_m'\|_{L^2(\Jd)}$ are clearly bounded.
With regard to \eqref{eimx1.dec.3}, \eqref{Um.est.1}, and \eqref{Um.est.2},
every function
$R_m(x_1,x_2):=\phi_m(x_1)U_m(x_2)$
belongs to~$W^{1,2}_{\rm per}(\Omega_0)$.

Our goal is to show that $R:=\sum_{m\in \Z} R_m$ is in
$W^{1,2}_{\rm per}(\Omega_0)$ as well.

The finite number of bounded terms with $m\leq m_0$
is included in the following estimates and equalities without any other specific comments.
The identity~\eqref{eimx1.dec.3} and inequality~\eqref{Um.est.1}
together with Fubini's theorem imply
\begin{equation*}
\left\| \sum_{m\in \Z} R_m \right\| \leq C_2 \|\Psi\|.
\end{equation*}
A similar estimate can be obtained for
$\partial_2 R_m $ provided that we use the inequality~\eqref{Um.est.2}.
For $\partial_1 R_m$, we have
\begin{equation*}
\left\|\sum_{m=-N}^N \partial_1 R_m \right\|^2 = \sum_{m=-N}^N m^2 \|U_m\|_{L^2(\Jd)}^2 \leq C_1^2 \sum_{m=-N}^N \frac{m^2}{m^2+1} \|\psi_m\|^2_{L^2(\Jd)},
\end{equation*}
where we used the inequality \eqref{Um.est.1}.
The fraction in the sum on the right hand side is bounded,
therefore, using the Parseval equality,
the limit $\sum_{m \in \Z} \partial_1 R_m$ remains in~$L^2(\Omega_0)$.
We conclude that~$R$ belongs to
$W^{1,2}(\Omega_0)$ and
\begin{equation*}
\|R\|_{W^{1,2}(\Omega_0)} \leq C_3 \|\Psi\|_{L^2(\Omega_0)}.
\end{equation*}

It remains to verify that $R$ belongs to $W^{1,2}_{\rm per}(\Omega_0)$.
We introduce the partial sum $R_N:=\sum_{m=-N}^N R_m$.
The fact that $R_N \in W^{1,2}_{\rm per}(\Omega_0)$ for every $N \in \N$ and
the (trace) embedding of $W^{1,2}(\Omega_0)$ in $L^2(\partial \Omega_0)$ yields
\begin{align*}
\big|\big(\varphi,R(-l,\cdot)-R(l,\cdot)\big)_{+}\big|
&= \big|\big(\varphi,R(-l,\cdot)-R_N(-l,\cdot)+R_N(-l,\cdot)-R(l,\cdot)\big)_{+}\big|
\\
&\leq 2 \,C_4 \,\|\varphi\|_{+} \, \|R-R_N\|_{W^{1,2}(\Omega_0)}
\end{align*}
for every $\varphi \in L^2(\Jd,\dd \nu_{(+1)})$;
$C_4$ is a constant depending only on $\Omega_0$.
Notice that the left hand side does not depend on~$N$.
Hence, the periodicity of $R$ is justified
by taking the limit $N\rightarrow +\infty$
and considering the arbitrariness of $\varphi$.

Now, knowing that~$R$ belongs to $W^{1,2}_{\rm per}(\Omega_0)$,
one can easily check that
\begin{equation*}
\forall \varphi \in W^{1,2}_{\rm per}(\Omega_0)\,, \qquad
 h_{\II}(\varphi,R)-z (\varphi,R)_{L^2(\Omega_0,G)}
 = (\varphi,\Psi)_{L^2(\Omega_0,G)} \,.
\end{equation*}
This implies that $R\in \Dom(H_{\II(+1)})$, see Lemma \ref{Dom.Hcs}, and $(H_{\II(+1)}-z)R=\Psi$.
\end{proof}
Proposition \ref{decomposition} has the important consequence for the spectrum of~$H_{\iota(K)}$.

\begin{corollary}\label{Corol.Davies}
$$
  \sigma\big(H_{\iota(K)}\big)
  = \bigcup_{m \in \mathbb{Z}} \sigma\big(H_{\iota(K)}^m\big)
$$
\end{corollary}

\begin{proof}
The inclusion $\displaystyle \sigma\big(H_{\iota(K)}\big) \subset \bigcup_{m \in \mathbb{Z}} \sigma\big(H_{\iota(K)}^m\big)$ follows from Proposition \ref{decomposition}, the other one is trivial.
\end{proof}

\begin{remark}
Notice that the statement of Corollary~\ref{Corol.Davies}
relating the spectra of a direct sum of operators
with their individual spectra does not hold in general
(\cf~\cite[Thm.~8.1.12]{Davies_2007}).
In our case, however, we have been able to prove the result
due to the compactness of resolvents
and additional information about the behaviour
of the numerical ranges of~$H_{\iota(K)}^m$ (\cf~Lemma~\ref{Num.range}).
\end{remark}

\subsection{Similarity to self-adjoint or normal operators}

We proceed with an analysis of $H_{\iota(K)}^m$.
For sake of simplicity, we drop the subscript~2 of the $x_2$ variable in the sequel.
We remark that \mbox{$\PT$-symmetry}
and \mbox{$\P$-pseudo-Hermiticity} of~$H_{\iota(K)}^m$ is preserved
with~$\P$ and~$\T$
being naturally restricted to~$L^2(\Jd,\dd \nu_{(K)})$.

The operators $H_{\iota(K)}^m$ are neither self-adjoint nor normal,
nevertheless we can show the following general result:
\begin{theorem}\label{Riesz.basis}
For every $m \in \Z$ and $K\in\{-1,0,1\}$:
\begin{enumerate}
\item The families of operators $H_{\II(K)}^m(\alpha,\beta),
$ $H_{\I(K)}^m(b, c, \phi)$ are holomorphic with respect to parameters
$\alpha, \beta,$ and $b, c, \phi$  entering the boundary conditions.
\item The spectrum of $H_{\iota(K)}^m$ is discrete consisting of simple eigenvalues
(\ie, the algebraic multiplicity being one),
except of finitely many eigenvalues of algebraic multiplicity two
and geometric multiplicity one that can appear
for particular values of $\alpha, \beta$ and $b, c, \phi$.
\item If all the eigenvalues are simple, then
\begin{enumerate}
\item[$a)$] the eigenvectors of $H_{\iota(K)}^m$ form a Riesz basis,
\item[$b)$] $H_{\iota(K)}^m$ is similar to a normal operator,
\ie, there exists a bounded operator $\varrho$
with bounded inverse such that $\varrho H_{\iota(K)}^m\varrho^{-1}$ is normal,
\item[$c)$] if moreover all eigenvalues are real,
then $H_{\iota(K)}^m$ is similar to a self-adjoint operator,
\ie, $\varrho H_{\iota(K)}^m\varrho^{-1}$ is self-adjoint.
\end{enumerate}
\item Let us denote by
$\big\{\psi_{i,m}\big\}_{i\in \N}$
the eigenfunctions of $H_{\iota(K)}^m$.
The set of eigenfunctions
$\mathscr{B}:=\big\{\phi_m \psi_{i,m} \big\}_{m\in \Z,i\in \N}$,
where $\phi_m$ were introduced in \eqref{eimx1.dec.2},
forms a Riesz basis of $L^2(\Omega_0,G)$.
\end{enumerate}
\end{theorem}
\begin{proof}
\emph{Ad~1.}
In view of \cite[Sect.~VII, Ex.~1.15]{Kato},
the Hamiltonians $H_{\II(K)}^m(\alpha,\beta)$,
considered as a family of operators depending
on parameters $\alpha, \beta$ entering boundary conditions, are holomorphic.
The same is true for $H_{\I(K)}^m(b, c, \phi)$.

\emph{Ad~2.}
The separated boundary conditions belong to the class of strongly regular boundary
conditions \cite{Naimark1967-LDOP1,Naimark1968-LDOP2}.
The connected $\PT$-symmetric boundary conditions are strongly regular
as well because $\theta_1=-b$, $\theta_{-1}=b$ (in Naimark's notation)
and~$b$ is non-zero by the assumption in~\eqref{BC}.
Moreover, all the eigenvalues are simple~\cite{Mikhajlov1962-3}
up to finitely many degeneracies that can appear:
eigenvalues with algebraic multiplicity two and geometric multiplicity one.

\emph{Ad~3.}
With regard to the strong regularity of boundary conditions,
the eigenfunctions of the Hamiltonian $H_{\iota(K)}^m$
form a Riesz basis \cite{Mikhajlov1962-3},
except the situations when the degeneracies appear.
The existence of Riesz basis implies the similarity
to a normal operator and as a special case the similarity
to a self-adjoint operator if the spectrum of $H_{\iota(K)}^m$ is real.

In more details, let $\big\{\psi_n\big\}_{n\in \N}$
be the Riesz basis of eigenvectors of $H_{\iota(K)}^m$,
\ie, $H_{\iota(K)}^m\psi_n=\lambda_n \psi_n$.
By definition, there exists a bounded operator~$\rho$
with bounded inverse such that $\big\{\rho \psi_n\big\}_{n\in \N}$
is an orthonormal basis that we denote by $\big\{e_n\big\}_{n\in \N}$.
Then
\begin{equation*}
\rho H_{\iota(K)}^m\rho^{-1}=\sum_{n\in\N} \lambda_n \, e_n (e_n,\cdot)_{L^2(\Omega_0,G)}
\end{equation*}
is a normal (self-adjoint if every $\lambda_n \in \R$) operator.

\emph{Ad~4.}
At first we show that $\mathscr{B}$ is complete, \ie,
$\mathscr{B}^{\perp}=\{0\}$.
Take $\omega \in \mathscr{B}^{\perp}$, \ie, for every $m\in\Z$, $i\in\N$,
\begin{eqnarray*}
0&=&\int_{\Omega_0}\overline{\phi_m(x_1) \psi_{i,m}(x_2)} \omega(x_1,x_2) \dd x_1 \dd \nu_{(K)}(x_2)  \nonumber \\
&=&\int_{\Jd} \overline{\psi_{i,m}}  \omega_m(x_2) \dd \nu_{(K)}(x_2),
\end{eqnarray*}
where
$\displaystyle \omega_m(x_2)
:=\int_{\Jj} \overline{\phi_m(x_1)} \omega(x_1,x_2) \dd x_1$.
Since $\big\{\psi_{i,m}\big\}_{i\in \N}$ forms a Riesz basis,
$\omega_m = 0$ a.e.\ in $L^2(\Jd,\dd \nu_{(K)})$
for every $m \in \Z$.
Since $\big\{\phi_m \big\}_{m\in\Z}$ is the~orthonormal basis of
$L^2(\Jj)$, $\omega=0$ in $L^2(\Omega_0,G)$.

Now we define an involution $(\P_1 \psi)(x_1,x_2):=\psi(-x_1,x_2)$.
We show that $\psi_{i,m}$
can be normalized in such way that $\mathscr{B}$ is $\P_1\T$-orthonormal,
\ie,
\begin{equation*}
\big(\phi_m\psi_{i,m},\P_1\T\phi_n\psi_{j,n}\big)_{L^2(\Omega_0,G)}
=\delta_{ij}\delta_{mn}.
\end{equation*}
Since $\P_1 \T \phi_m=\phi_m$, $\P_1\T$-orthogonality follows immediately
for $m\neq n$ because
$\phi_m$ are orthonormal in $(\cdot,\cdot)_{L^2(\Jj)}$ and $G$ is independent of~$x_1$.
For $m=n$ we have
\begin{equation}\label{P1T.og}
\big(\phi_m\psi_{i,m}, \P_1\T \phi_m \psi_{j,m}\big)_{L^2(\Omega_0,G)}
=\big(\psi_{i,m},\T \psi_{j,m}\big)_{L^2(\Jd,\dd \nu_{(K)})}
\,.
\end{equation}
If $i \neq j$, then the right hand side of \eqref{P1T.og} is zero
because $\T \psi_{j,m}$ is an~eigenfunction of $\big(H_{\iota(K)}^m\big)^*$.
Indeed, it is a general fact that
eigenfunctions of~$H$ and~$H^*$ corresponding to different eigenvalues
are orthogonal.
It remains to verify that if $i=j$, then the right hand side of \eqref{P1T.og}
does not vanish, \ie,
\begin{equation*}
\int_{\Jd}\psi_{j,m}^2(x_2) \, \dd \nu_{(K)}(x_2) \neq 0.
\end{equation*}
However, this is precisely the condition on $\lambda_{j,m}$
being a simple eigenvalue of $H_{\iota(K)}^m$.
It can be easily seen either directly
or it follows from \cite[Thm.~5]{Garcia-2008-179}.
\end{proof}

\begin{remark}
Notice that an additional symmetry with respect to $\P_1$ was essential in the proof.
The set of eigenfunctions $\mathscr{B}$ is not $\T$-orthonormal because the products
$\big(\phi_m\psi_{i,m},\T \phi_m \psi_{i,m}\big)_{L^2(\Omega_0,G)}$ vanish. This situation is typical for \mbox{$\T$-self-adjoint} operators with eigenvalues that are not simple \cite{Garcia-2008-179}.
\end{remark}

\subsection{Separated boundary conditions}
\label{SSsepBC}

At first, we investigate the Hamiltonians $H_{\II(0)}^m(\alpha,\beta)$.
Then $H_{\II(\pm 1)}^m(\alpha,0)$ are analysed.
These results together allow us to describe the remaining $\beta \neq 0$ case.

\subsubsection{Zero curvature}

As expected, the zero curvature case is the simplest
and it will serve as a reference model.
In fact, the corresponding one-dimensional eigenvalue problem
\begin{equation}\label{EV.zero}
\left\{
 \begin{aligned}
-\psi'' + m^2 \psi &= k^2 \psi \quad {\rm in\quad } (-a,a), \\
 \psi'(\pm a)+ (\ii \alpha \pm \beta)\, \psi(\pm a)&=0,
  \end{aligned}
\right.
\end{equation}
has been already studied previously in \cite{krejcirik-2006-39}.
Here we overtake the main results.
\begin{proposition}
The spectrum of $H_{\II(0)}^m(\alpha,0)$ is real for all $m\in\Z$.
The eigenvalues
$\lambda_{j,m}$ and eigenfunctions $\psi_{j,m}$
can be written in the following form, $m \in \Z$,
\begin{eqnarray}\label{H0EV}
\lambda_{j,m} &=& \begin{cases}
 \alpha^2+m^2& \mbox{if}\quad j=0 \,,
    \\
 k_j^2+m^2&\mbox{if}\quad j \geq 1 \, ,
\end{cases} \\
\psi_{j,m}(x) &=&
  \begin{cases}\displaystyle
    C_{0} \, \exp{(-\ii \alpha x)}
    & \mbox{if}\ j=0 \,,
    \\ \displaystyle
    C_{j} \left(
    \cos (k_j x) + \frac{k_j \sin (k_j a) - \ii \alpha \cos (k_j a)}{k_j \cos (k_j a) + \ii \alpha \sin (k_j a)} \, \sin (k_j x)
    \right)
    &\mbox{if}\ j \geq 1,
  \end{cases} \nonumber
\end{eqnarray}
where $k_j:= \frac{j\pi}{2a}$.
If $\alpha^2 \neq k_j^2$, \ie,
there is no level-crossing for the same $m$, then the operator is similar to a self-adjoint operator or, equivalently, it is \mbox{quasi-Hermitian}.
\end{proposition}
\begin{remark}
Closed formulae for the metric operator $\Theta$ for $H_{\II(0)}^m(\alpha,0)$
are presented in \cite{krejcirik-2006-39,krejcirik-2008-41a}.
The similarity transformation $\varrho$ can be found as $\varrho=\sqrt{\Theta}$
or as any other decomposition of the positive operator $\Theta=\varrho^*\varrho$.
\end{remark}
The $\alpha$-dependence of eigenvalues $\lambda$ for $m=0,1,2$ is plotted in Figure \ref{SEVcylinder}.
\begin{figure}
    \centering
        \includegraphics[width=0.60\textwidth]{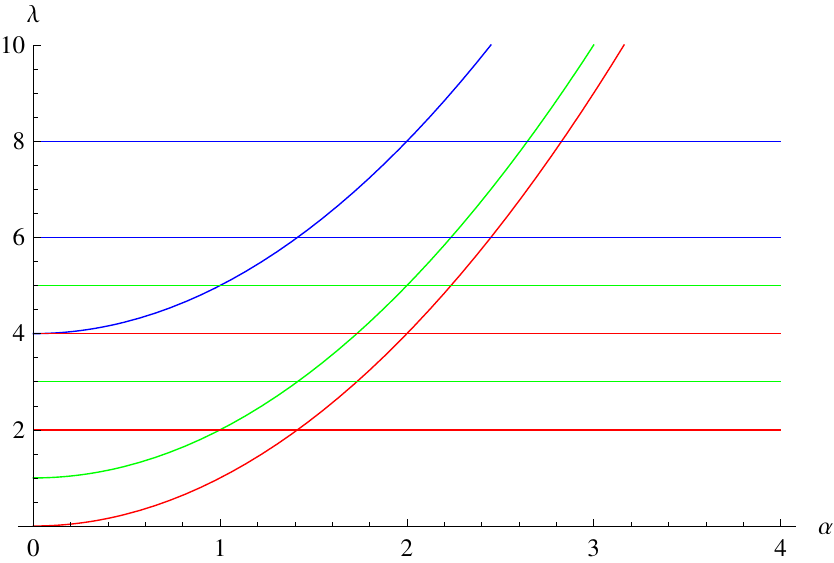}
    \caption{$\alpha$-dependence of eigenvalues, zero curvature, $a=\pi/4$. Red, green, and blue colour corresponds to $m=0,1,2$ respectively.}
    \label{SEVcylinder}
\end{figure}

The case of $\beta\neq 0$ is more complicated
and as it was remarked in \cite{krejcirik-2006-39},
the spectrum of $H_{\II(0)}^m(\alpha,\beta)$ can be complex.
More precise results follow from a further analysis,
not presented in~\cite{krejcirik-2006-39}.

\begin{proposition}
\
\begin{enumerate}
\item
If $\beta>0$, then the spectrum of $H_{\II(0)}^m(\alpha,\beta)$ is purely real for all $m\in\Z$ and $\alpha\in\R$.
\item
If $\beta<0$, then the spectrum of $H_{\II(0)}^0(\alpha,\beta)$
is either purely real or there is one pair of complex conjugated eigenvalues
with real part located in the neighbourhood of $\alpha^2+\beta^2$.
For fixed negative $\beta$, the points $\alpha_n$ where a pair of complex conjugated eigenvalues appears
(by increasing of $\alpha$) are determined by $\alpha_n^2+\beta^2=k_n^2$, where \mbox{$k_n^2:=((2n+1)\pi/4a)^2$}
for some $n\in \N$.
\end{enumerate}
\noindent
The eigenvalues $\lambda=k^2$ of $H_{\II(0)}^0(\alpha,\beta)$
are determined ($k= 0$ is admissible only if $\alpha=\beta=0$) by the equation
\eq{ (k^2-\alpha^2-\beta^2)\sin (2ka)-2\beta k\cos (2ka)=0.  \label{EVbeta}}
The corresponding eigenfunctions read
\eq{\psi(x) = C \left( \cos (k x) + \frac{k \sin (k a) - (\ii \alpha + \beta) \cos (k a)}{k \cos (k a)
+ (\ii \alpha + \beta) \sin (k a)} \, \sin (k x)   \right).    \label{H0EF}}
The eigenvalues of $H_{\II(0)}^m(\alpha,\beta)$ are obtained by adding $m^2$ to the eigenvalues of $H_{\II(0)}^0(\alpha,\beta)$.

\label{ThmSpCyl}
\end{proposition}
\begin{proof}
We proceed in a similar way as in the alternative
proof \cite[Sect.~6.1]{krejcirik-2006-39}
of the reality of the spectrum of $H_{\II(0)}^0(\alpha,0)$.
The original eigenvalue problem \eqref{EV.zero} with $m=0$
can be transformed, using $\phi(x):=e^{\ii \alpha x} \psi(x) $, into
\begin{equation}\label{CTrEVp}
\left\{
\begin{aligned}
-\phi''+2\ii\alpha\phi'+\alpha^2\phi &=\lambda \phi \quad {\rm in\quad } (-a,a), \\
\phi'(\pm a) \pm \beta\phi(\pm a)&=0.
\end{aligned}
\right.
\end{equation}
We multiply the equation \eqref{CTrEVp} by $\overline{\phi''}$
and integrate over $(-a,a)$.
Next we multiply the complex conjugated version of the equation~\eqref{CTrEVp}
by $\phi''$ and again we integrate over $(-a,a)$. By subtracting the results and integrating by parts with use of the boundary conditions in \eqref{CTrEVp}, we obtain the identity
\eq{- \alpha \beta^2 \left( |\phi(a)|^2-|\phi(-a)|^2    \right) =
\Im \lambda \left( \|\psi'\|^2_{L^2(\Jd)}+\beta (|\phi(a)|^2+|\phi(-a)|^2) \right).  \label{PrReal1} }
If we perform the same procedure, however, with multiplication by $\overline{\phi}$, after some integration by parts we receive the relation
\eq{\alpha \left( |\phi(a)|^2-|\phi(-a)|^2    \right) = \Im \lambda \, \|\phi\|^2_{L^2(\Jd)}. \label{PrReal2} }
Combining \eqref{PrReal1} with \eqref{PrReal2} leads to the identity
\eq{0 = \Im \lambda \left( \|\phi'\|^2_{L^2(\Jd)}+
\beta (|\phi(a)|^2+|\phi(-a)|^2) +\beta^2 \|\phi\|^2_{L^2(\Jd)} \right). \label{PrReal3}}
If $\beta$ is positive, then the whole term in the brackets is strictly positive and thus imaginary part of $\lambda$ must be zero.
This proves the first item of the proposition.

If $\beta$ is negative, then complex eigenvalues can appear. If we divide the equation \eqref{EVbeta} by $k^2$ and leave only $\sin (2ka)$ term on the left hand side, then it is clear that eigenvalues approach $(n \pi/2a)^2$ for $k$ real and large enough. After simple algebraic manipulation \eqref{EVbeta} becomes
\eq{\tan (2ka)=\frac{2 \beta k}{k^2-\alpha^2-\beta^2}  \label{EVbeta2}}
and the eigenvalues correspond to the intersections of the graphs of functions on left and right hand side of \eqref{EVbeta2}.
We denote $l(k)$ the function on the left hand side,
$r(k)$ the one on the right hand side, and $k_0:=\sqrt{\alpha^2+\beta^2}$.
The behaviour of $r(k)$ for $k\in\R$ is summarized in Table \ref{r.behaviour}.
\begin{table}[ht]
\begin{center}
\begin{tabular}{|c|c|c|}
\hline
$k$                      & sign & asymptotics \\ \hline \hline
$(-\infty, -k_0)$ &  $r(k)>0,\ r'(k)>0,\ r''(k)>0$ & $ \lim_{k\rightarrow -{k_0}_- }r(k)=+\infty$\\
                        &                              &  $  \lim_{k\rightarrow -\infty}r(k)=0$ \\ \hline
$(-k_0, 0)$ & $r(k)<0,\ r'(k)>0,\  r''(k)<0$ & $\lim_{k\rightarrow -{k_0}_+ }r(k)=-\infty $\\ \hline
$(0,k_0)$   & $r(k)>0,\ r'(k)>0,\  r''(k)>0$ & $\lim_{k\rightarrow  {k_0}_- }r(k)= +\infty $\\ \hline
$(k_0,\infty)$ &     $r(k)<0,\ r'(k)>0,\ r''(k)<0$ & $\lim_{k\rightarrow {k_0}_- }r(k)=-\infty$ \\
                        &                              &       $\lim_{k\rightarrow \infty}r(k)=0$ \\ \hline
\end{tabular}
\caption{The behaviour of $r(k)$.}
\label{r.behaviour}
\end{center}
\end{table}

%
Graphs of functions $l(k)$ and $r(k)$ are plotted in Figure \ref{GraphBetaN}.
\begin{figure}[ht]
\begin{center}
\subfloat[$\alpha=0$]{\includegraphics[width =0.45\textwidth]{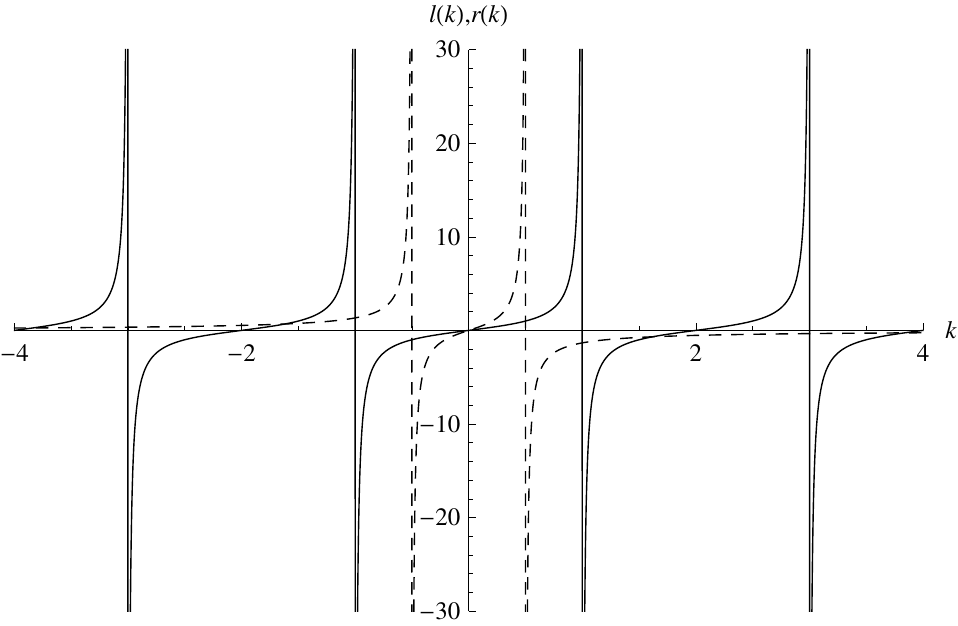}} \hspace{0.3cm}
\subfloat[$\alpha=3.1$]{\includegraphics[width =0.45\textwidth]{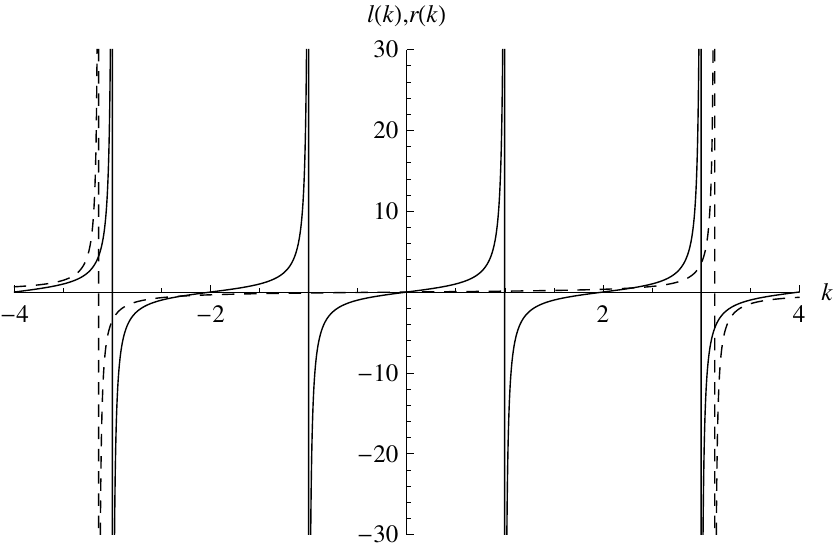}}
\caption{ Graphs of $l(k)$ (full line) and $r(k)$ (dashed line), $a=\pi/4$, $\beta=-0.5$.}
\label{GraphBetaN}
\end{center}
\end{figure}
It is clear from the holomorphic dependence of eigenvalues on $\alpha, \beta$
(a consequence of Theorem \ref{Riesz.basis})
that eigenvalues are close to $(n \pi/a)^2$,
corresponding to zeros of $l(k)=\tan (2ka)$, except those in the neighbourhood of $\alpha^2+\beta^2$.
Since $\alpha=\beta=0$ case corresponds to Neumann boundary conditions, for small $\alpha$ and $\beta$, all eigenvalues must be close to $(n \pi/2a)^2$.
Hence, if we fix $\beta$ and increase $\alpha$, then two intersections of graphs of $l(k)$ and $r(k)$ are ``lost"
precisely at the point where $\alpha_n^2+\beta^2=k_n^2$ for some $n\in\N$,
\ie, the asymptote of $r(k)$ corresponds to the asymptote of the tangent $l(k)$.
This implies the creation of complex conjugate pair of eigenvalues.
If we increase $\alpha$ more, two intersections appear again which means the annihilation of complex
conjugate pair, \ie, the restoration of two real eigenvalues. The two intersections are lost at the next critical value $\alpha_{n+1}$. Very rough estimates give the location of restoration of real eigenvalues in the interval
$(n\pi/2a,(2n+1)\pi/4a)$.

In view of the presented arguments, only one complex conjugated pair can appear in the
spectrum for fixed $\alpha$ and $\beta$ in the neighbourhood of $\alpha^2+\beta^2$,
and the other eigenvalues approach $(n \pi/2a)^2$ as the distance from $\alpha^2+\beta^2$ increases.
Moreover, for fixed $\beta$, the enlarging of $\alpha$
results into the shift of eigenvalues from almost Neumann ones $(n \pi/2a)^2, n\in \N$,
to Dirichlet ones $((n+1) \pi/2a)^2, n\in \N,$ for $\alpha$ large.

Finally, the equation for eigenvalues and eigenfunctions are found in a standard way.
The general solution $A\cos (kx) +B \sin (k x)$ of $-\psi''=k^2\psi$ is inserted into boundary conditions
\eqref{EV.zero} and the condition for existence of non-trivial solutions $A,B$ is
the eigenvalue equation \eqref{EVbeta}.
\end{proof}

Figures \ref{SEVCylBetaP}, \ref{SEVCylBetaN} represent
the $\alpha$-dependence of the first four eigenvalues
as obtained by a numerical analysis of~\eqref{EVbeta}.
The numerical results confirm the above described behaviour.
Let us remark that if $\beta$ is positive, then the graph of $r(k)$
is reflected by the $x$-axis and the effect of loosing intersections is not possible,
hence the spectrum remains real.
\begin{figure}[ht]
\begin{center}
\includegraphics[width=0.60\textwidth]{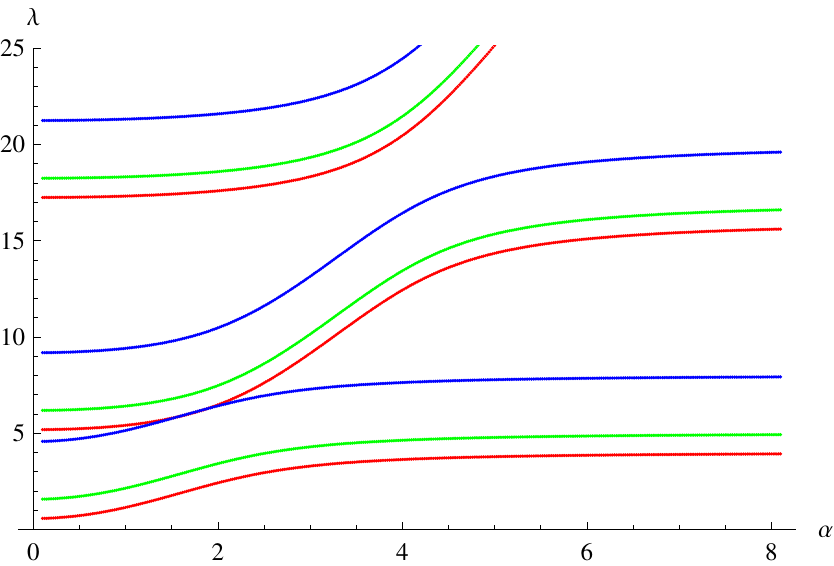}
\caption {$\alpha$-dependence of eigenvalues, zero curvature, $a=\pi/4$, $\beta=0.5$.
Red, green, and blue colour corresponds to $m=0,1,2$ respectively.}
\label{SEVCylBetaP}
\end{center}
\end{figure}
\begin{figure}[ht]
\begin{center}
\subfloat[Real part of $\lambda$]{\includegraphics[width =0.45\textwidth]{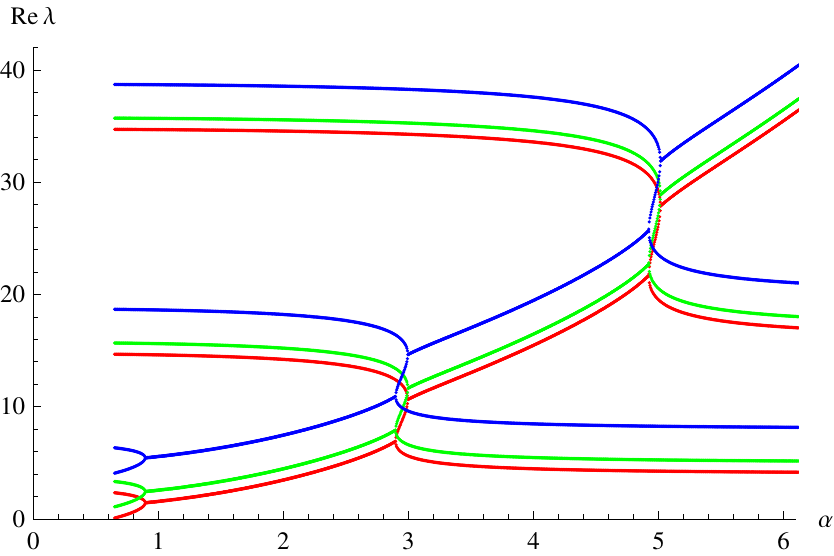}}
\subfloat[Imaginary part of $\lambda$]{\includegraphics[width =0.45\textwidth]{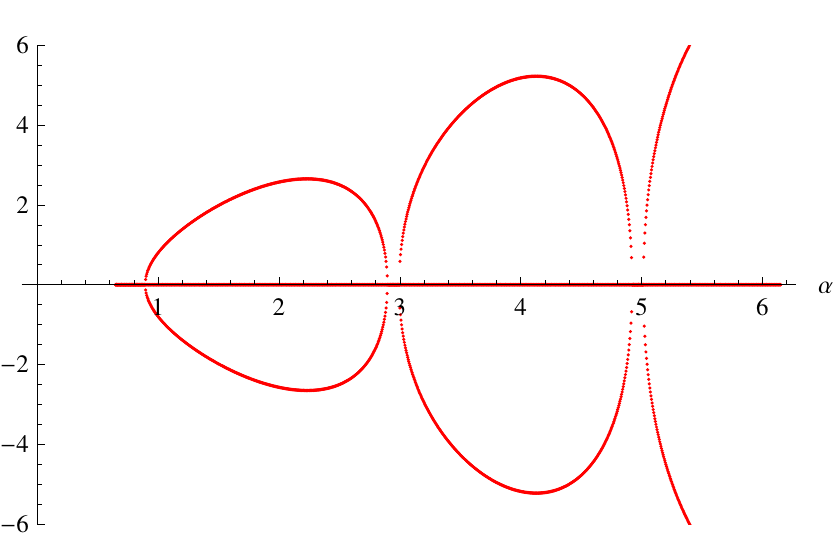}}
\caption {$\alpha$-dependence of eigenvalues, zero curvature, $a=\pi/4$, $\beta=-0.5$.
Red, green, and blue colour corresponds to $m=0,1,2$ respectively.}
\label{SEVCylBetaN}
\end{center}
\end{figure}

\subsubsection{Positive curvature}

The eigenvalue problem for the Hamiltonian $H_{\II(+1)}^m$ reads
\begin{equation}\label{SEVeq}
\left\{
\begin{aligned}
\displaystyle -\psi''(x) + \tan x \, \psi' (x) +\frac{m^2}{\cos^2x}\psi(x)
&=k^2 \psi(x) \quad {\rm in\quad } (-a,a), \\
\psi'(\pm a)+\ii \alpha \psi(\pm a)&=0.
\end{aligned}
\right.
\end{equation}
Solutions of \eqref{SEVeq} can be written down in terms of associated Legendre functions
$P^{(\mu)}_{\nu},Q^{(\mu)}_{\nu}$:
\eq{\psi(x)=C_1 \psi_1(x)+ C_2\psi(x) \equiv C_1 P^{(m)}_{\nu}(\sin x) + C_2 Q^{(m)}_{\nu}(\sin x),\label{SGS}}
where
\eq{\nu:=\frac{1}{2}\left(\sqrt{1+4 \lambda}-1 \right),   }
\eq{C_2 (\alpha  \psi_2(-a)-\ii \psi_2'(-a)) =C_1 (-\alpha  \psi_1(-a)+\ii \psi_1'(-a)) .\label{C2SSphere}}
Inserting the general solution \eqref{SGS} into boundary conditions in \eqref{SEVeq}
and consequent search for non-trivial constants $C_1,C_2$ yields the eigenvalue equation
\eq{\left|
\begin{array}{ll}
\psi_1'(a)+\ii \alpha \psi_1(a) &  \psi_2'(a)+\ii \alpha \psi_2(a) \\
\psi_1'(-a)+\ii \alpha \psi_1(-a) & \psi_2'(a)+\ii \alpha \psi_2(a)
\end{array} \right|=0. \label{SIEVeq}}

In order to analyse the spectrum in more details,
we transform the Hamiltonian $H_{\II(+1)}^m$ into a unitarily equivalent operator
of a more convenient form. The proof of the lemma is a straightforward calculation.
\begin{lemma}\label{U.eq.pos}
The unitary mapping $U_{(+1)}: L^2(\Jd,\dd x) \rightarrow L^2(\Jd,\dd\nu_{(+1)})$
\eq{\left(U_{(+1)}\psi\right)(x):= (\cos x )^{-\frac{1}{2}}  \,  \psi(x)\label{STrS}}
transforms $H_{\II(+1)}^m(\alpha,0)$ to
\eq{U_{(+1)}^{-1} H_{\II(+1)}^m(\alpha,0)U_{(+1)}
= H_{\II(0)}^0(\alpha,\mbox{$\frac{1}{2}$} \tan a )+V_{(+1)}^m, \label{TrHp} }
where
\eq{V_{(+1)}^m(x):=\frac{8 m^2-3-\cos 2x}{8 \cos^2 x}. \label{Vmp}}
\end{lemma}

Equipped with the equivalent form of the Hamiltonian,
we prove the following result.
\begin{proposition}
For every $m\in\Z$ there exists a real number $\Lambda_{(+1)}^m$
such that all eigenvalues $\lambda$ with $\Re \lambda \geq \Lambda_{(+1)}^m$
are real and simple
(\ie\ the algebraic multiplicity being one).
The eigenvalues with $\Re \lambda < \Lambda_{(+1)}^m$ can be complex, ordered in complex conjugated pairs.

Eigenvalues are determined by equation \eqref{SIEVeq}
and eigenfunctions can be written in the form \eqref{SGS}
with~\eqref{C2SSphere}.
\label{SpectrumS}
\end{proposition}
\begin{proof}
Let us consider the transformed Hamiltonian \eqref{TrHp}
and forget about the potential for a moment,
\ie, we understand the potential $V_{(+1)}^m$
as a perturbation of~$H_{\II(0)}^0(\alpha,\frac{1}{2}\tan a )$.
Since $\tan a$ is positive under the assumption $a < \pi/2$,
the reality of the spectrum is guaranteed by Proposition \ref{ThmSpCyl}.1.
The potential represents a bounded perturbation and it can shift
eigenvalues only by $C \|V_{(+1)}^m\|$.
Here the constant~$C$ comes from the estimate of the norm of the resolvent
$$\|R_{\II(0)}^0(\lambda)\|\leq \frac{C}{\Im \lambda} \,,$$
 which is valid for $H_{\II(0)}^0(\alpha,\frac{1}{2} \tan a )$
due to the similarity to a normal operator
(\cf~Theorem \ref{Riesz.basis}).
The separation distance $|\lambda_{n+1}-\lambda_n|$ of eigenvalues (ordered with respect to the real part)
of the unperturbed
operator $H_{\II(0)}^0(\alpha,\frac{1}{2} \tan a )$ grows to infinity and two eigenvalues must collide at first
to create a complex conjugate pair. Hence, the perturbed operator cannot have more than finitely many complex
eigenvalues. Recall that due to \mbox{$\PT$-symmetry}
(Corollary \ref{Complex.pairs})
the complex eigenvalues come in complex conjugated pairs.
\end{proof}
\begin{remark}\label{rem.pos}
In other words, we detected the effects of positive curvature. It~acts as the adding of real bounded potential $V_{(+1)}^m$ and  real ``$\beta$ like" term in the boundary conditions to the zero curvature Hamiltonian $H_{\II(0)}^0$. The positive $\frac{1}{2} \tan a$ term is decisive for the behaviour of the spectrum, the bounded potential $V_{(+1)}^m$ can affect substantially only the lowest eigenvalues. Nonetheless, we conjecture that the spectrum remain real for every $m\in \Z$.
\end{remark}

A numerical analysis of the equation~\eqref{SIEVeq}
for $\lambda=k^2$ is presented in Figure~\ref{SEVsphere}.
Obvious similarity with Figure~\ref{SEVCylBetaP} supports
the perturbative results.
\begin{figure}
    \centering
        \includegraphics[width=0.60\textwidth]{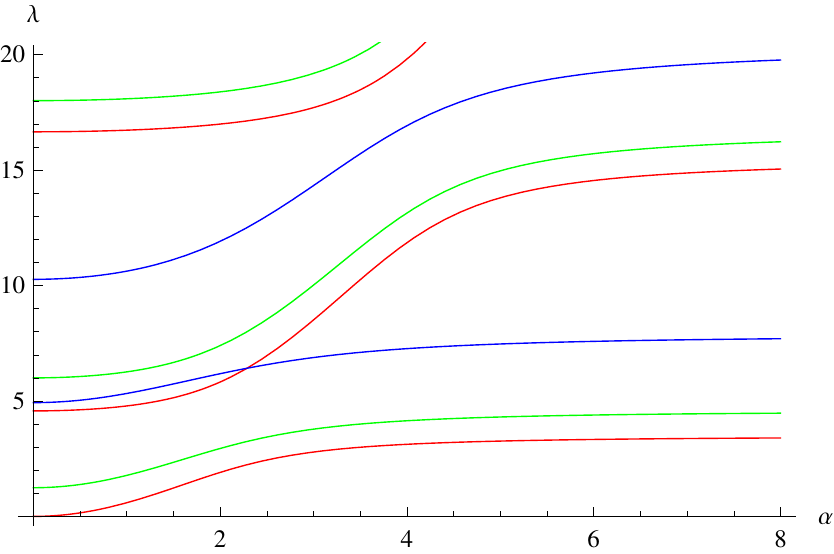}
    \caption{$\alpha$-dependence of eigenvalues, positive curvature, $a=\pi/4$. Red, green, and blue colour corresponds to $m=0,1,2$ respectively.}
    \label{SEVsphere}
\end{figure}

\subsubsection{Negative curvature}

The eigenvalue problem for the Hamiltonian $H_{\II(-1)}^m$ reads
\begin{equation}\label{PSEVeq}
\left\{
\begin{aligned}
-\psi''(x)-\tanh x\, \psi'(x) +\frac{m^2}{\cosh^2x} \psi(x)
&= k^2 \psi(x) \quad {\rm in \quad } (-a,a), \\
\psi'(\pm a)+\ii \alpha \psi(\pm a) &= 0.
\end{aligned}
\right.
\end{equation}
The solutions of~\eqref{PSEVeq} can be again
expressed via associated Legendre functions $P^{(\mu)}_{\nu},Q^{(\mu)}_{\nu}$,
but they have a little bit more complicated form then \eqref{SGS}:
\eq{\psi(x)=C_1 \psi_1(x)+ C_2\psi(x) \equiv C_1 \frac{P^{(\mu)}_{\nu}(\tanh x)}{\sqrt{\cosh x}}
+ C_2 \frac{Q^{(\mu)}_{\nu}(\tanh x)}{\sqrt{\cosh x}},\label{PSGS}}
where
\eq{\mu:=\ii m-\frac{1}{2}, \ \ \nu:=\frac{1}{2}\sqrt{1-4 \lambda}.   }
Relations between $C_1,C_2$ can be obtained from equation \eqref{C2SSphere}, however, with $\psi_1, \psi_2$ corresponding to the negative curvature solutions \eqref{PSGS};
the same is true for the eigenvalue equation \eqref{SIEVeq}.

To explain the behaviour of the spectrum in a deeper way,
we use the same strategy as in the positive curvature case.
The eigenvalue problem \eqref{PSEVeq} can be transformed by an analogous
unitary transformation leading to a modified zero curvature eigenvalue problem.
\begin{lemma}\label{TrPs}
The unitary mapping $U_{(-1)}: L^2(\Jd,\dd x) \rightarrow  L^2(\Jd,\dd\nu_{(-1)})$
\eq{\left(U_{(-1)}\psi\right)(x):=  (\cosh x) ^{-\frac{1}{2}} \, \psi(x)\label{STrPS}}
transforms $H_{\II(-1)}^m(\alpha,0)$ to
\eq{U_{(-1)}^{-1} H_{\II(-1)}^m(\alpha,0)U_{(-1)}= H_{\II(0)}^0(\alpha,-\frac{1}{2} \tanh a )+V_{(-1)}^m, \label{TrHm} }
where
\eq{V_{(-1)}^m(x):=\frac{8 m^2+3+\cosh 2x}{8 \cosh^2x}. \label{Vmm} }
\end{lemma}
\begin{proposition}
For every $m\in\Z$ there exists a real number $\Lambda_{(-1)}^m$
such that all eigenvalues $\lambda$ with $\Re \lambda \geq \Lambda_{(-1)}^m$
are either real and simple (\ie~the algebraic multiplicity being one),
or there is one complex conjugated pair of eigenvalues with real part located in the neighbourhood of $\alpha^2+\beta^2$. The eigenvalues with $\Re \lambda < \Lambda_{(-1)}^m$ can be complex, ordered in complex conjugated pairs.

Eigenvalues are determined by equation \eqref{SIEVeq} with $\psi_1, \psi_2$ from \eqref{PSGS}.
Eigenfunctions can be written in the form \eqref{PSGS} with constants $C_1,C_2$ satisfying \eqref{C2SSphere} with $\psi_1, \psi_2$ from \eqref{PSGS}.
\end{proposition}
\begin{proof}
The proof is the same as in the positive curvature case,
\cf~the proof of Proposition \ref{SpectrumS}.
The unperturbed Hamiltonian $H_{\II(0)}^0(\alpha,-\frac{1}{2} \tanh a )$
corresponds to the case analysed in Proposition \ref{ThmSpCyl}.2.
\end{proof}

\begin{remark}\label{rem.neg}
The curvature effect is now represented by the bounded real potential $V_{(-1)}^m$ and the extra negative term $-\frac{1}{2} \tanh a$ in the boundary conditions.
\end{remark}

A result of the numerical analysis of the eigenvalue problem
is presented in Figure \ref{SEVpseudo}.
The resemblance to zero curvature case with negative~$\beta$
in boundary conditions is obvious.

\begin{figure}
\begin{center}
\subfloat[ Real part of $\lambda$ ]{
\includegraphics[ width =0.45\textwidth ]{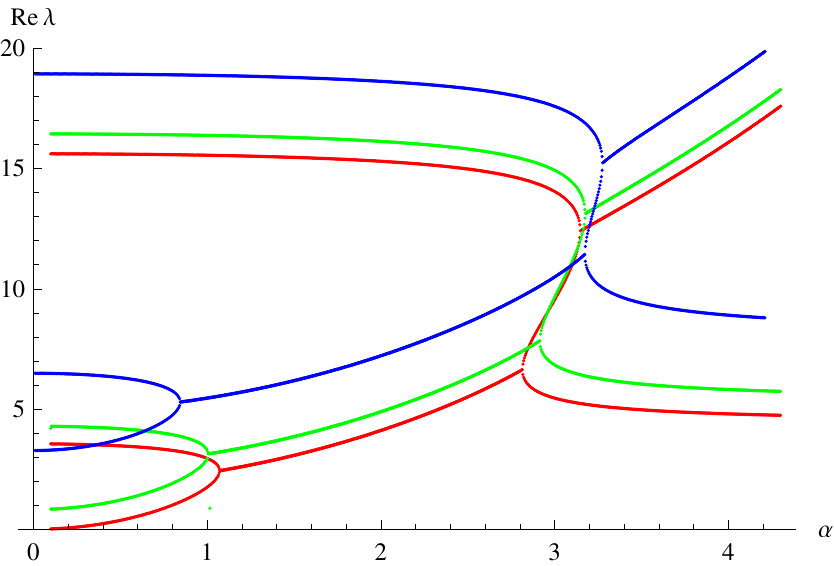}
}\subfloat[ Imaginary part of $\lambda$ ]{
\includegraphics [width =0.45\textwidth ]{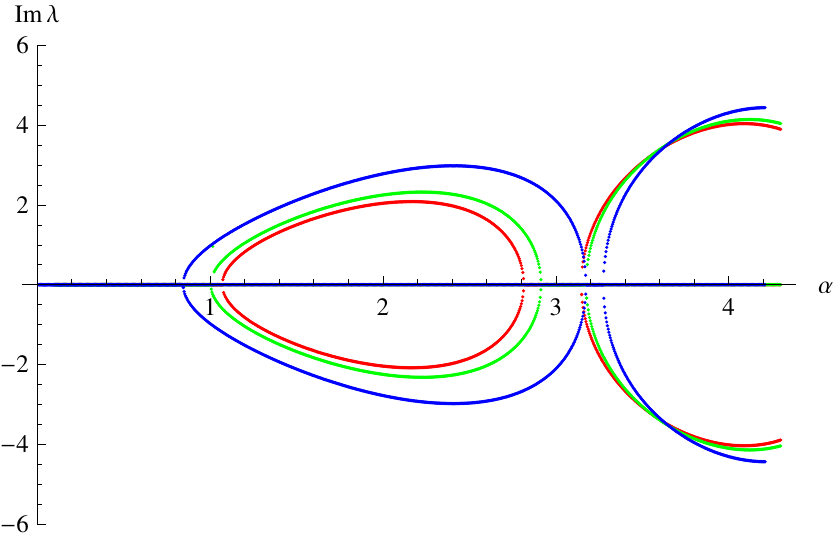}}
\caption {$\alpha$-dependence of eigenvalues, negative curvature, $a=\pi/4$. Red, green, and blue colour corresponds to $m=0,1,2$ respectively. See \href{http://gemma.ujf.cas.cz/~siegl/PTCurvedMan.html}{animation}, for an animated visualization of the $\alpha$-dependence of the eigenvalues.}
\label{SEVpseudo}
\end{center}
\end{figure}

\subsection{Connected boundary conditions}

The connected boundary conditions are, by their nature,
more complicated than the separated ones and moreover,
they are given by three real parameters $b,c,\phi$.
Like for the separated boundary conditions,
we can use the unitary transformations $U_{(\pm 1)}$
introduced in \eqref{STrS}, \eqref{STrPS}
to transform the problems to the zero curvature case,
however, with modified boundary conditions
and with additional bounded real potentials $V_{(\pm 1)}^m$
defined in \eqref{Vmp}, \eqref{Vmm}.
The modification of boundary conditions is presented in appropriate subsections below.

The spectrum is not analytically described so far even for the zero curvature model
and it is beyond the scope of this article to proceed with this analysis.
The main aim of this section is to show the effect of curvature,
\ie, the transformation of curved models to the zero curvature case.
Furthermore, we present some results of a numerical analysis
for the `lowest'
eigenvalues: $\phi$-dependence for selected values of $b, c$.
It is important to note that, unlike in the separated case,
we do not start with our parameters $b,c,\phi$
from a self-adjoint operator for $\phi=0$,
as it was the case for $\alpha =0$ in the case of separated boundary conditions.
We remark that the case \mbox{$b=c=0$}, $\phi=\pm \pi/2$
corresponds to irregular boundary conditions and the spectrum
of such operators is completely different
from the cases presented here (\cf~\cite{Siegl2009-06}).

\subsubsection{Zero curvature}

We impose connected boundary conditions (\ref{BC}\I) on the solutions of eigenvalue problem for $H_{\I(0)}^0(b,c,\phi)$ and we obtain the following equation for eigenvalues $\lambda=k^2$
\eq{-2 k+2 k \cos (2ak) \sqrt{1+b c}\, \cos\phi+\left(b k^2-c\right) \sin (2ak) =0 \label{EVeqConZ}}
and eigenfunctions
\eq{ \psi(x)=C_1 \cos (k x) + C_2\sin (kx),   \label{EFConZ} }
where the constants are further restricted by
\begin{multline}\label{EFConZC}
C_2 \left( \left(-1+\sqrt{1+b c} e^{\ii \phi }\right) \cos (ak) +b k \sin (a k) \right)
\\
= C_1 \left( \left(1+\sqrt{1+b c} e^{\ii \phi }\right) \sin (ak) -b k \cos (ak) \right).
\end{multline}
\begin{proposition}
Eigenvalues $\lambda=k^2$ of $H_{\I(0)}^0(b,c,\phi)$ are determined by equation \eqref{EVeqConZ}, eigenfunctions read \eqref{EFConZ} with \eqref{EFConZC}.

Eigenvalues for $m \neq 0$ can be obtained by the shift $\lambda \mapsto \lambda + m^2$ while the corresponding eigenfunctions remain the same.
\end{proposition}

Figure \ref{CEVaCylinder}
illustrates the behaviour of eigenvalues for a certain choice of parameters.

\begin{figure}
\begin{center}
\subfloat[ Real part of $\lambda$ ]{
\includegraphics[ width =0.45\textwidth ]{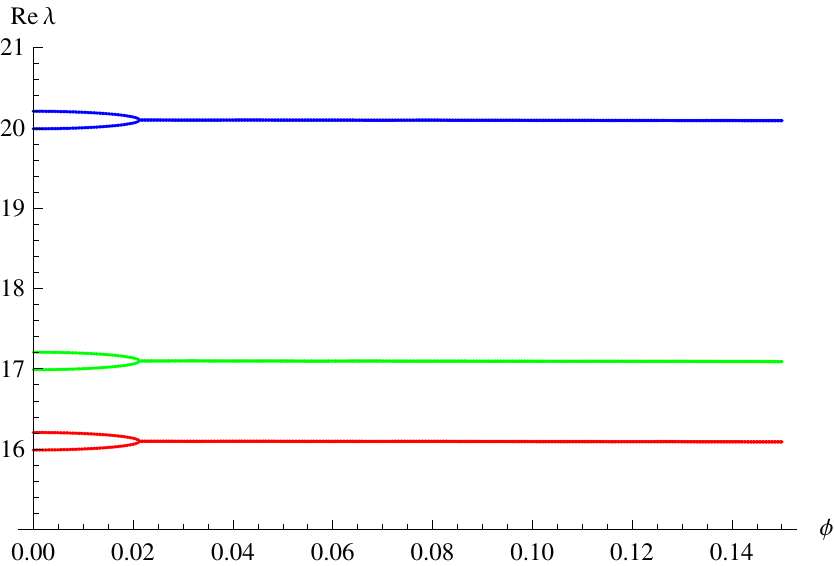}
}\subfloat[ Imaginary part of $\lambda$ ]{
\includegraphics [width =0.45\textwidth ]{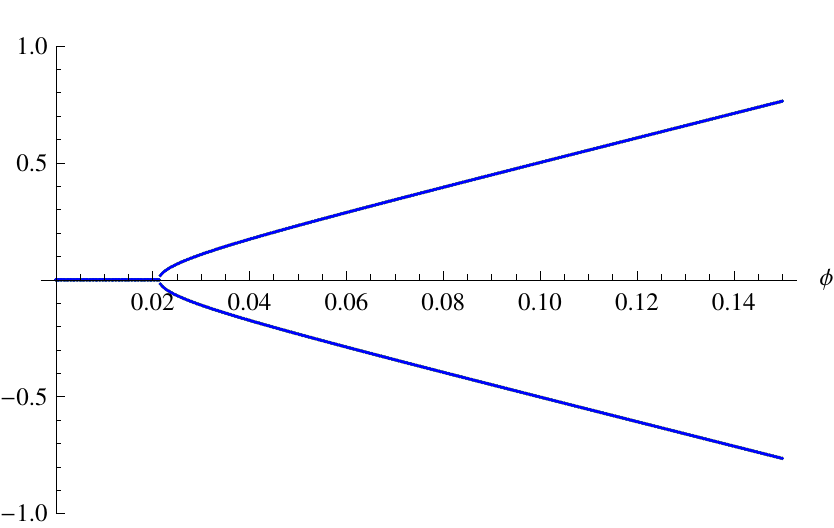}}
\caption {$\phi$-dependence of eigenvalues, zero curvature, $a=\pi/4$, $b=c=0.01$. Red, green, and blue colour corresponds to $m=0,1,2$ respectively. }
\label{CEVaCylinder}
\end{center}
\end{figure}

\subsubsection{Positive curvature}

The solutions of the eigenvalue problem for $H_{\I(+1)}^m(b,c,\phi)$ with connected boundary conditions (\ref{BC}\I) are the same as \eqref{SGS} except the constants $C_1, C_2$ now satisfy
\begin{multline}\label{C2CPSphere}
C_2\left( \sqrt{1+b c} e^{\ii \phi } \psi_2(-a)-\psi_2(a)+b \psi_2'(-a) \right)
\\
= C_1  \left(  \sqrt{1+b c} e^{\ii \phi } \psi_1(-a)-\psi_1(a)+b \psi_1'(-a)  \right).
\end{multline}
The equation for eigenvalues reads
\begin{equation}\label{SCEVeq}
\left|
\begin{array}{cc}
 -\sqrt{1+b c}\, e^{\ii \phi } \psi_1(-a)+\psi_1(a)-b \psi_1'(-a) & -\sqrt{1+b c}\, e^{\ii \phi } \psi_2(-a)+\psi_2(a)-b \psi_2'(-a) \\
 -c \psi_1(-a)-\sqrt{1+b c}\, e^{-\ii \phi } \psi_1'(-a)+\psi_1'(a) & -c \psi_2(-a)-\sqrt{1+b c}\, e^{-\ii \phi } \psi_2'(-a)+\psi_2'(a)
\end{array}
\right|=0.
\end{equation}

Figure~\ref{CEVaSphere}
illustrates the behaviour of eigenvalues for a certain choice of the parameters.
\begin{figure}
\begin{center}
\subfloat[ Real part of $\lambda$ ]{
\includegraphics[ width =0.45\textwidth ]{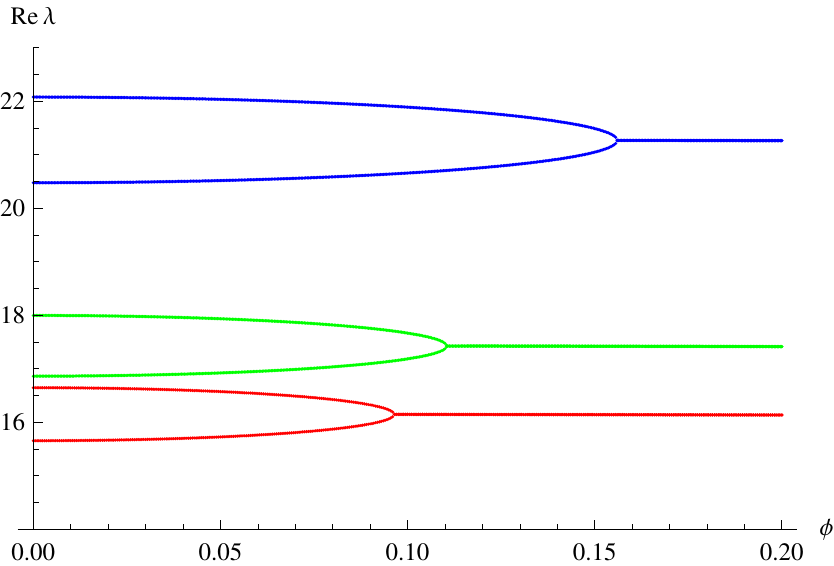}
}\subfloat[ Imaginary part of $\lambda$ ]{
\includegraphics [width =0.45\textwidth ]{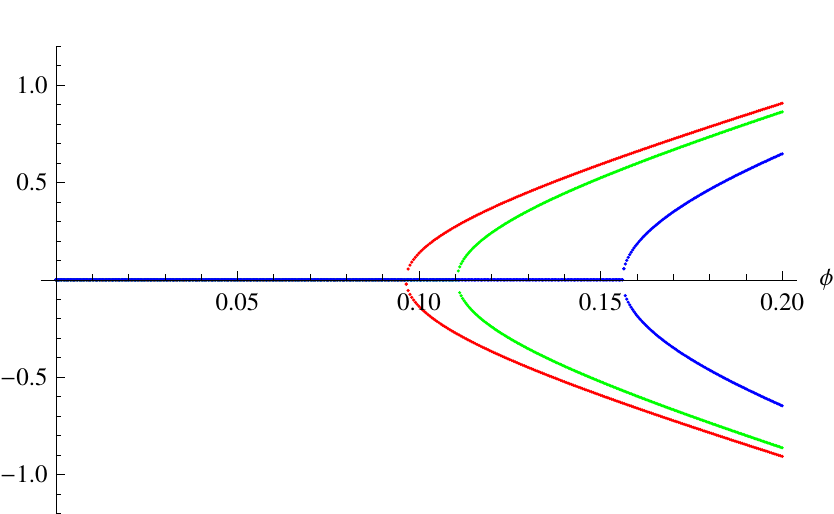}}
\caption {$\phi$-dependence of eigenvalues, positive curvature, $a=\pi/4$, \mbox{$b=c=0.01$}. Red, green, and blue colour corresponds to $m=0,1,2$ respectively. }
\label{CEVaSphere}
\end{center}
\end{figure}

We employ the unitary transformation $U_{(+1)}$
introduced in Lemma~\ref{U.eq.pos}
to map $H_{\I(+1)}^m(b,c,\phi)$ to a zero curvature Hamiltonian.

\begin{proposition}
The unitary mapping $U_{(+1)}$ defined in~\eqref{STrS}
transforms the Hamiltonian $H_{\I(+1)}^m(b,c,\phi)$ to
\begin{equation}\label{TrHpCon}
U_{(+1)}^{-1} H_{\I(+1)}^m(b,c,\phi)U_{(+1)}= \hat{H}_{\I(0)}^0 +V_{(+1)}^m,
\end{equation}
where $V_{(+1)}^m$ is defined in \eqref{Vmp} and
$\displaystyle \hat{H}_{\I(0)}^0:=-\frac{\dd^2}{\dd x^2}$ with the domain
consisting of~$\psi \in W^{2,2}(\Jd)$ satisfying
\begin{eqnarray}
\Psi(a)&=&B_{(+1)} \Psi(-a), \ \ {\rm with} \ \label{ModBCS}
\Psi(x):= \matice{\psi(x) \\ \psi'(x)} \ {\rm and}  \\
B_{(+1)}&:=&\matice{\sqrt{1+bc)}\, e^{\ii\phi}-\frac{1}{2}b\tan a & b \\ c-\sqrt{1+bc} \tan a \cos\phi +
\frac{1}{4} b \tan^2 a & \sqrt{1+bc}\,e^{-\ii\phi}-\frac{1}{2}b\tan a}. \nonumber
\end{eqnarray}

Eigenvalues $\lambda=k^2$ of $H_{\I(+1)}^m(b,c,\phi)$ are determined by equation \eqref{SCEVeq}, eigenfunctions
read \eqref{SGS} with constants $C_1,C_2$ given by \eqref{C2CPSphere}.
\end{proposition}
\begin{remark}
The boundary conditions \eqref{ModBCS} are $\PT$-symmetric, but they are no more $\P$-pseudo-Hermitian.
This result shows that although we reduced the problem to the zero curvature case
(in the sense of previous sections), the investigation of spectrum must be done with more
general boundary conditions than $\PT$-symmetric and $\P$-pseudo-Hermitian at the same time.
\end{remark}

\subsubsection{Negative curvature}

The solutions of the eigenvalue problem for $H_{\I(-1)}^m(b,c,\phi)$ with connected boundary conditions (\ref{BC}\I)
are the same as in the separated conditions case \eqref{PSGS}, but the relation between constants $C_1, C_2$ is
given by \eqref{C2CPSphere} with $\psi_1, \psi_2$ corresponding to the negative curvature solutions \eqref{PSGS};
the same is also valid for the eigenvalue equation \eqref{SCEVeq}.

Figure \ref{CEVaPSphere}
illustrates the behaviour of eigenvalues for a certain choice of parameters.
\begin{figure}
\begin{center}
\subfloat[ Real part of $\lambda$ ]{
\includegraphics[ width =0.45\textwidth ]{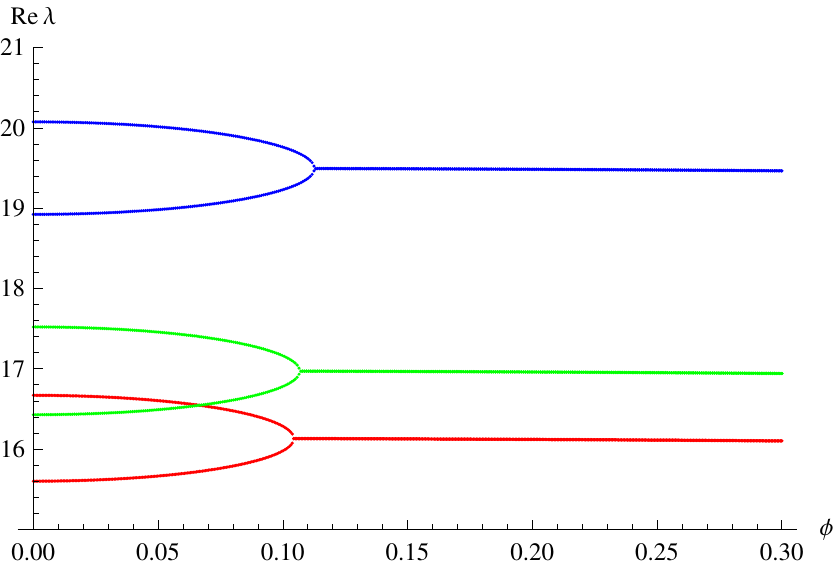}
}\subfloat[ Imaginary part of $\lambda$ ]{
\includegraphics [width =0.45\textwidth ]{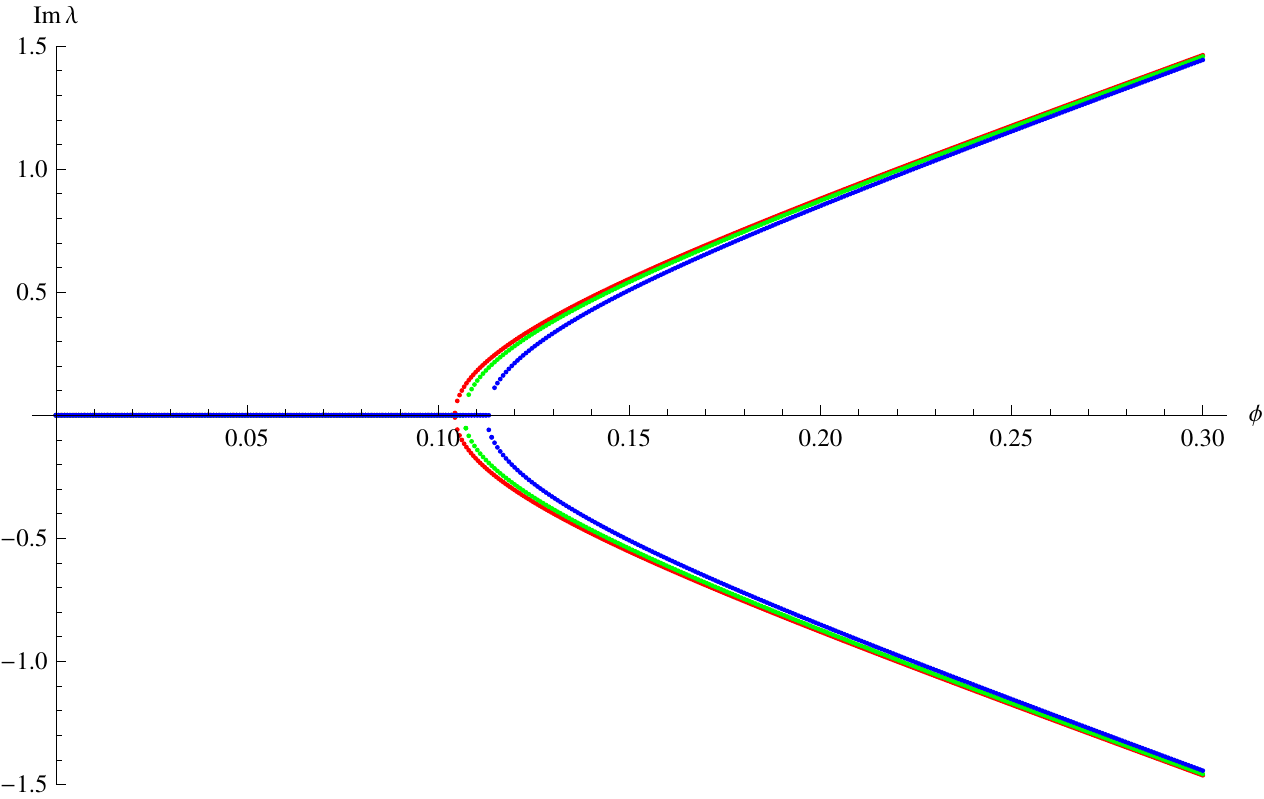}}
\caption {$\phi$-dependence of eigenvalues, negative curvature, $a=\pi/4$, \mbox{$b=c=0.01$}.
Red, green, and blue colour corresponds to $m=0,1,2$ respectively. }
\label{CEVaPSphere}
\end{center}
\end{figure}

\begin{proposition}
The unitary mapping $U_{(-1)}$ defined in~\eqref{STrPS}
transforms the Hamiltonian $H_{\I(-1)}^m(b,c,\phi)$ to
\eq{U_{(-1)}^{-1} H_{\I(-1)}^m(b,c,\phi)U_{(-1)}= \tilde{H}_{\I(0)}^0 +V_{(-1)}^m, \label{TrHmCon} }
where $V_{(-1)}^m(x)$ is defined in \eqref{Vmm} and $\displaystyle \tilde{H}_{\I(0)}^0:=-
\frac{\dd^2}{\dd x^2}$ with the domain consisting of $\psi \in W^{2,2}(\Jd)$ satisfying
\begin{eqnarray}
\Psi(a)&=&B_{(-1)} \Psi(-a), \ \ {\rm with} \ \Psi(x):= \matice{\psi(x) \\ \psi'(x)} \ {\rm and}
\label{ModBCPS} \\
B_{(-1)}&:=&\matice{\sqrt{1+bc)}\, e^{\ii\phi}+\frac{1}{2}b\tanh a & b \\ c+\sqrt{1+bc} \tanh a \cos\phi +
\frac{1}{4} b \tanh^2 a & \sqrt{1+bc}\,e^{-\ii\phi}+\frac{1}{2}b\tanh a}. \nonumber
\end{eqnarray}

Eigenvalues $\lambda=k^2$ of $H_{\I(-1)}^m(b,c,\phi)$ are determined by equation \eqref{SCEVeq}
with $\psi_{1},\psi_2$ from \eqref{PSGS}. The eigenfunctions read \eqref{PSGS}, where constants $C_1,C_2$
are given by \eqref{C2CPSphere} with $\psi_{1},\psi_2$ from \eqref{PSGS}.
\end{proposition}
\begin{remark}
The boundary conditions \eqref{ModBCPS} are $\PT$-symmetric, however not \mbox{$\P$-pseudo-Hermitian},
as for the positive curvature case. Thus again, it is necessary to investigate more general
boundary conditions in zero curvature eigenvalue problem.
\end{remark}

\section{Concluding remarks}\label{Sec.end}

The goal of this paper was to introduce
a new class of $\PT$-symmetric Hamiltonians defined in curved manifolds
and describe the effects of curvature on the spectrum.
Although we were able to find these effects for both separated
and connected boundary conditions,
the absence of results on reality of the spectrum for the latter
(even in the case of zero curvature)
did not allow us to present the conclusions in an entirely descriptive
and explicit way.
Let us therefore summarize the main features of the model
for the separated Robin type boundary conditions~\eqref{Robin} only.

In Table~\ref{Tab.summary} we schematically (and very roughly)
describe qualitative properties of the spectrum we observed
in the constant-curvature cases.
\begin{table}[ht]
\begin{center}
\begin{tabular}{|c|c|l|}
\hline
curvature   &  spectrum     & eigenvalues \\ \hline \hline
zero            & $\R$      & only some $\alpha$-dependent, crossings   \\ \hline
positive    & $\R$          & all $\alpha$-dependent, no crossings       \\  \hline
negative    & $\C$          & all $\alpha$-dependent, crossings, \\
                &               & creation and annihilation of complex pairs \\ \hline
\end{tabular}
\caption{A heuristic summary of our analytical and numerical analysis.}\label{Tab.summary}
\end{center}
\end{table}
The entry describing the positive curvature case
includes our conjecture (supported by numerical analysis)
that all eigenvalues are real.

One of the most instructive results in the paper are probably
Lemmata~\ref{U.eq.pos} and~\ref{TrPs}, which enable one
to understand the effect of curvature in terms of
an additional effective potential and boundary-coupling interaction.
For the $s$-wave modes (\ie~$m=0$ in the decomposition~\eqref{H.dec})
and infinitesimally thin strips (\ie~$a \ll l$),
it follows from the lemmata
that the positive and negative curvature acts
as an attractive and repulsive interaction, respectively.
This is in agreement with a spectral analysis of similar models
in the self-adjoint case of Dirichlet boundary conditions
\cite{Krejcirik-2003-45,Krejcirik-2006-2006}.
However, the additional boundary interaction is not negligible
for positive widths~$a$, and its effect is actually completely opposite
(\cf~Remarks \ref{rem.pos}, \ref{rem.neg}):
the positive and negative curvature gives rise to
an attractive and repulsive Robin-type boundary condition, respectively.
The interplay between these two effects is further complicated
by the presence of the repulsive centrifugal term for $|m| \geq 1$,
and the numerical analysis confirms that the overall picture
of the spectrum can be quite complex.

It follows from previous comments and remarks that there remain several open problems,
\eg~the proof of the reality of all eigenvalues in the positive curvature model.
Nonetheless, we would like to mention also some other interesting directions
of potential future research:
the spectral effect of curvature in non-constant curvature and non-constant boundary-coupling functions setting,
the existence of Riesz basis for such setting or models defined on unbounded domains (waveguides) in curved
spaces. The last case can be viewed as a natural continuation of \cite{borisov-2007}
where a planar $\PT$-symmetric waveguide was studied.

\section*{Acknowledgement}
P.S. is thankful to B.~Mityagin for very valuable discussions.
The work was partially supported by the Czech Ministry of Education,
Youth and Sports within the project LC06002.
P.S. appreciates also the support of CTU grant No.~CTU0910114.

\addcontentsline{toc}{section}{References}
{\footnotesize
\bibliographystyle{acm}
\bibliography{references}

\begin{thebibliography}{10}

\bibitem{Adams1975}
{\sc Adams, R.~A.}
\newblock {\em Sobolev spaces}.
\newblock Academic Press, New York, 1975.

\bibitem{albeverio-2002-59}
{\sc Albeverio, S., Fei, S.~M., and Kurasov, P.}
\newblock {Point {I}nteractions: PT-Hermiticity and Reality of the Spectrum}.
\newblock {\em Letters in Mathematical Physics 59\/} (2002), 227--242.

\bibitem{Albeverio-2009-42}
{\sc Albeverio, S., Gunther, U., and Kuzhel, S.}
\newblock {J-self-adjoint operators with $\mathcal{C}$-symmetries: an extension
  theory approach}.
\newblock {\em Journal of Physics A: Mathematical and Theoretical 42}, 10
  (2009), 105205 (22pp).

\bibitem{Albeverio2005-38}
{\sc Albeverio, S., and Kuzhel, S.}
\newblock {One-dimensional Schr\"{o}dinger operators with $\mathcal
  P$-symmetric zero-range potentials}.
\newblock {\em Journal of Physics A: Mathematical and General 38}, 22 (2005),
  4975--4988.

\bibitem{Andrianov-2009-5}
{\sc Andrianov, A.~A., Bender, C.~M., Jones, H.~F., Smilga, A., and Znojil,
  M.}, Eds.
\newblock {\em {Quantum Physics with Non-Hermitian Operators}\/} (2009),
  vol.~5, SIGMA.

\bibitem{Bendali-1996-56}
{\sc Bendali, A., and Lemrabet, K.}
\newblock {The effect of a thin coating on the scattering of a time-harmonic
  wave for the Helmholtz equation}.
\newblock {\em SIAM Journal on Applied Mathematics 56}, 6 (1996), 1664--1693.

\bibitem{Bender2007-70}
{\sc Bender, C.~M.}
\newblock {Making sense of non-Hermitian Hamiltonians}.
\newblock {\em Reports on Progress in Physics 70}, 6 (2007), 947--1018.

\bibitem{bender-1998-80}
{\sc Bender, C.~M., and Boettcher, S.}
\newblock {Real Spectra in Non-Hermitian Hamiltonians Having PT Symmetry}.
\newblock {\em Physical Review Letters 80\/} (1998), 5243--5246.

\bibitem{bender2002-89}
{\sc Bender, C.~M., Brody, D.~C., and Jones, H.~F.}
\newblock {Complex Extension of Quantum Mechanics}.
\newblock {\em Physical Review Letters 89\/} (2002), 270401 (4pp).

\bibitem{borisov-2007}
{\sc Borisov, D., and {Krej\v ci\v r\'ik}, D.}
\newblock {PT-symmetric waveguides}.
\newblock {\em Integral Equations Operator Theory 62}, 4 (2008), 489--515.

\bibitem{Bouchitte-1989-24}
{\sc Bouchitt{\'e}, G., and Petit, R.}
\newblock {On the concepts of a perfectly conducting material and of a
  perfectly conducting and infinitely thin screen}.
\newblock {\em Radio Science 24\/} (1989), 13--26.

\bibitem{Boulton-Levitin-Marletta}
{\sc Boulton, L., Levitin, M., and Marletta, M.}
\newblock {A PT-symmetric periodic problem with boundary and interior
  singularities}.
\newblock {\em Journal of Differential Equations\/} (2009).

\bibitem{Chugunova-2008-342}
{\sc Chugunova, M., and Pelinovsky, D.}
\newblock {Spectrum of a non-self-adjoint operator associated with the periodic
  heat equation}.
\newblock {\em Journal of Mathematical Analysis and Applications 342}, 2
  (2008), 970--988.

\bibitem{Clark-1996-29}
{\sc Clark, I.~J., and Bracken, A.~J.}
\newblock {Effective potentials of quantum strip waveguides and their
  dependence upon torsion}.
\newblock {\em Journal of Physics A: Mathematical and General 29}, 2 (1996),
  339--348.

\bibitem{Davies1995}
{\sc Davies, E.~B.}
\newblock {\em {Spectral theory and differential operators}}.
\newblock Cambridge University Press, 1995.

\bibitem{Davies_2007}
{\sc Davies, E.~B.}
\newblock {An indefinite convection-diffusion operator}.
\newblock {\em LMS Journal of Computation and Mathematics 10\/} (2007),
  288--306.

\bibitem{Davies-2007}
{\sc Davies, E.~B.}
\newblock {\em {Linear operators and their spectra}}.
\newblock Cambridge University Press, 2007.

\bibitem{Engquist-1993}
{\sc Engquist, B., and Nedelec, J.-C.}
\newblock {Effective boundary conditions for electromagnetic scattering in thin
  layers}.
\newblock Rapport interne CMAP 278, 1993.

\bibitem{Evans1998}
{\sc Evans, L.~C.}
\newblock {\em Partial {D}ifferential {E}quations}.
\newblock AMS, Providence, 1998.

\bibitem{Fring-2008-41}
{\sc Fring, A., Jones, H., and Znojil, M.}, Eds.
\newblock {\em {Pseudo-Hermitian Hamiltonians in Quantum Physics VI}\/} (2008),
  vol.~41, Journal of Physics A: Mathematical and Theoretical.

\bibitem{Garcia-2008-179}
{\sc Garcia, S.~R.}
\newblock {The Eigenstructure of Complex Symmetric Operators}.
\newblock {\em Proceedings of the Sixteenth International Conference on
  Operator Theory and Applications (IWOTA 16) 179\/} (2008), 169--183.

\bibitem{Gilbarg-Trudinger}
{\sc Gilbarg, D., and Trudinger, N.~S.}
\newblock {\em Elliptic Partial Differential Equations of Second Order}.
\newblock Springer-Verlag, Berlin, 1983.

\bibitem{Gray}
{\sc Gray, A.}
\newblock {\em Tubes}.
\newblock Addison-Wesley Publishing Company, New York, 1990.

\bibitem{Hartman_1964}
{\sc Hartman, P.}
\newblock Geodesic parallel coordinates in the large.
\newblock {\em American Journal of Mathematics 86\/} (1964), 705--727.

\bibitem{Hebey}
{\sc Hebey, E.}
\newblock {\em Nonlinear Analysis on Manifolds: Sobolev Spaces and
  Inequalities}.
\newblock AMS, New York, 2000.

\bibitem{Jacob-2008-8}
{\sc Jacob, B., Trunk, C., and Winklmeier, M.}
\newblock {Analyticity and Riesz basis property of semigroups associated to
  damped vibrations}.
\newblock {\em Journal of Evolution Equations 8\/} (2008), 263--281.

\bibitem{jain-2009-73}
{\sc Jain, S.~R., and Ahmed, Z.}, Eds.
\newblock {\em {Non-Hermitian Hamiltonians in Quantum Physics}\/} (2009),
  vol.~73, Pramana - Journal Of Physics.

\bibitem{Kaiser-2002-43}
{\sc Kaiser, H.-C., Neidhardt, H., and Rehberg, J.}
\newblock {Density and current of a dissipative Schr{\"o}dinger operator}.
\newblock {\em Journal of Mathematical Physics 43}, 11 (2002), 5325--5350.

\bibitem{Kaiser-2003-45}
{\sc Kaiser, H.-C., Neidhardt, H., and Rehberg, J.}
\newblock {Macroscopic current induced boundary conditions for
  Schr{\"o}dinger-type operators}.
\newblock {\em Integral Equations and Operator Theory 45}, 1 (2003), 39--63.

\bibitem{Kaiser-2003-252}
{\sc Kaiser, H.-C., Neidhardt, H., and Rehberg, J.}
\newblock {On 1-dimensional dissipative Schr\"odinger-type operators their
  dilations and eigenfunction expansions}.
\newblock {\em Mathematische Nachrichten 252\/} (2003), 51--69.

\bibitem{Kato}
{\sc Kato, T.}
\newblock {\em Perturbation theory for linear operators}.
\newblock Springer-Verlag, 1966.

\bibitem{Klaiman-2008-101}
{\sc Klaiman, S., {G\"{u}nther}, U., and Moiseyev, N.}
\newblock {Visualization of Branch Points in $\mathcal{PT}$-Symmetric
  Waveguides}.
\newblock {\em Physical Review Letters 101}, 8 (2008), 080402 (4pp).

\bibitem{Krejcirik-2003-45}
{\sc Krej{\v c}i{\v r}{\'i}k, D.}
\newblock {Quantum strips on surfaces}.
\newblock {\em Journal of Geometry and Physics 45\/} (2003), 203--217.

\bibitem{Krejcirik-2006-2006}
{\sc Krej{\v ci\v r\'i}k, D.}
\newblock {Hardy inequalities in strips on ruled surfaces}.
\newblock {\em Journal of Inequalities and Applications 2006\/} (2006), (10pp).

\bibitem{krejcirik-2008-41a}
{\sc Krej{\v ci\v r\'i}k, D.}
\newblock {Calculation of the metric in the Hilbert space of a {$\mathcal
  P$}{$\mathcal T$}-symmetric model via the spectral theorem}.
\newblock {\em Journal of Physics A: Mathematical and Theoretical 41}, 24
  (2008), 244012 (6pp).

\bibitem{krejcirik-2006-39}
{\sc Krej{\v c}i{\v r}{\'i}k, D., B{\'i}la, H., and Znojil, M.}
\newblock {Closed formula for the metric in the Hilbert space of a
  $\mathcal{PT}$-symmetric model}.
\newblock {\em Journal of Physics A: Mathematical and General 39}, 32 (2006),
  10143--10153.

\bibitem{krejcirik-2008-41}
{\sc Krej{\v c}i{\v r}{\'i}k, D., and Tater, M.}
\newblock {Non-Hermitian spectral effects in a $\mathcal{PT}$-symmetric
  waveguide}.
\newblock {\em Journal of Physics A: Mathematical and Theoretical 41}, 24
  (2008), 244013 (14pp).

\bibitem{Langer-2004-54}
{\sc Langer, H., and Tretter, C.}
\newblock {A Krein Space Approach to PT-symmetry}.
\newblock {\em Czechoslovak Journal of Physics 54\/} (2004), 1113--1120.

\bibitem{Mikhajlov1962-3}
{\sc Mikhajlov, V.}
\newblock {Riesz bases in ${\cal L}\sb 2 (0,1)$.}
\newblock {\em Sov. Math., Dokl., translation from Dokl. Akad. Nauk SSSR 114,
  981-984 (1962) 3\/} (1962), 851--855.

\bibitem{Mitchell-2001-63}
{\sc Mitchell, K.~A.}
\newblock Gauge fields and extrapotentials in constrained quantum systems.
\newblock {\em Physical Review A 63}, 4 (2001), 042112 (20pp).

\bibitem{Mostafazadeh-2002-43a}
{\sc Mostafazadeh, A.}
\newblock {Pseudo-Hermiticity versus PT-symmetry I: The necessary condition for
  the reality of the spectrum of a non-Hermitian Hamiltonian}.
\newblock {\em Journal of Mathematical Physics 43}, 1 (2002), 205--214.

\bibitem{Mostafazadeh-2002-43b}
{\sc Mostafazadeh, A.}
\newblock {Pseudo-Hermiticity versus PT-symmetry II: A complete
  characterization of non-Hermitian Hamiltonians with a real spectrum}.
\newblock {\em Journal of Mathematical Physics 43}, 5 (2002), 2814--2816.

\bibitem{Mostafazadeh-2002-43c}
{\sc Mostafazadeh, A.}
\newblock {Pseudo-Hermiticity versus PT-symmetry III: Equivalence of
  pseudo-Hermiticity and the presence of antilinear symmetries}.
\newblock {\em Journal of Mathematical Physics 43}, 8 (2002), 3944--3951.

\bibitem{Mostafazadeh2008-review}
{\sc Mostafazadeh, A.}
\newblock {Pseudo-Hermitian Quantum Mechanics}.
\newblock arXiv:0810.5643v2, 2008.

\bibitem{Naimark1967-LDOP1}
{\sc Naimark, M.}
\newblock {\em {Linear differential operators. Part I: Elementary theory of
  linear differential operator. Translated by E.R. Dawson. English translation
  edited by W.N. Everitt.}}
\newblock {New York: Frederick Ungar Publishing Co. XIII }, 1967.

\bibitem{Naimark1968-LDOP2}
{\sc Naimark, M.}
\newblock {\em {Linear differential operators. Part II: Linear differential
  operators in Hilbert space. Translated by E. R. Dawson. English translation
  edited by W. N. Everitt.}}
\newblock {New York: Frederick Ungar Publishing Co. XV }, 1968.

\bibitem{Rubinstein-2007-99}
{\sc Rubinstein, J., Sternberg, P., and Ma, Q.}
\newblock {Bifurcation Diagram and Pattern Formation of Phase Slip Centers in
  Superconducting Wires Driven with Electric Currents}.
\newblock {\em Physical Review Letters 99}, 16 (2007), 167003.

\bibitem{Rubinstein-2010-195}
{\sc Rubinstein, J., Sternberg, P., and Zumbrun, K.}
\newblock {The Resistive State in a Superconducting Wire: Bifurcation from the
  Normal State}.
\newblock {\em Archive for Rational Mechanics and Analysis 195\/} (2010),
  117--158.

\bibitem{Ruschhaupt-2005-38}
{\sc Ruschhaupt, A., Delgado, F., and Muga, J.~G.}
\newblock {Physical realization of $\mathcal{PT}$ -symmetric potential
  scattering in a planar slab waveguide}.
\newblock {\em Journal of Physics A: Mathematical and General 38}, 9 (2005),
  L171--L176.

\bibitem{Scholtz-1992-213}
{\sc Scholtz, F.~G., Geyer, H.~B., and Hahne, F. J.~W.}
\newblock {Quasi-Hermitian operators in quantum mechanics} and the variational
  principle.
\newblock {\em Annals of Physics 213\/} (1992), 74--101.

\bibitem{Siegl-MT}
{\sc Siegl, P.}
\newblock {Quasi-Hermitian Models}.
\newblock Master's thesis, Faculty of Nuclear Sciences and Physical
  Engineering, CTU Prague, http://ssmf.fjfi.cvut.cz/2008/siegl\_thesis.pdf,
  2007/2008.

\bibitem{Siegl2009-06}
{\sc Siegl, P.}
\newblock {Surprising spectra of $\PT$-symmetric point interactions}.
\newblock arXiv:0906.0226, June 2009.

\bibitem{Siegl-2009-73}
{\sc Siegl, P.}
\newblock {The non-equivalence of pseudo-Hermiticity and presence of antilinear
  symmetry}.
\newblock {\em Pramana - Journal of Physics 73}, 2 (2009), 279--287.

\bibitem{Spivak-1975}
{\sc Spivak, M.}
\newblock {\em A Comprehensive Introduction to Differential Geometry}, vol.~IV.
\newblock Publish or Perish, Boston, Mass., 1975.

\bibitem{Wachsmuth-2009}
{\sc Wachsmuth, J., and Teufel, S.}
\newblock {Effective Hamiltonians for Constrained Quantum Systems}.
\newblock arXiv:0907.0351, 2009.

\bibitem{Weir-2009-22}
{\sc Weir, J.}
\newblock {An indefinite convection-diffusion operator with real spectrum}.
\newblock {\em Applied Mathematics Letters 22}, 2 (2009), 280--283.

\end{thebibliography}
}

\end{document}